\documentclass[preprint,12pt]{elsarticle}

\biboptions{numbers,sort&compress}

\usepackage{amssymb}
\usepackage{amsmath}
\usepackage{amsthm}

\usepackage{bbold}
\usepackage{hyperref}
\usepackage{natbib}
\usepackage{subfig}
\usepackage{listings}
\usepackage{xcolor}

\journal{Chaos, Solitons \& Fractals}

\theoremstyle{thmstyleone}%
\newcommand\beq{\begin{equation}}%
\newcommand\eeq{\end{equation}}%
\newtheorem{defi}{Definition}%
\newtheorem{prop}{Proposition}%
\newtheorem{theo}{Theorem}%
\newtheorem{lem}{Lemma}%
\theoremstyle{remark}%
\newtheorem{rem}{Remark}%
\newcommand\bb{\mathbb}%
\DeclareMathOperator{\Int}{Int}%
\DeclareMathOperator{\Mod}{mod}%

\begin{document}

\begin{frontmatter}



\title{Chaotic discretization theorems for forced linear and nonlinear coupled oscillators}

\author[1]{Stefano Disca\corref{cor1}}
\ead{dscsfn@unife.it}
\author[1]{Vincenzo Coscia}
\ead{cos@unife.it}
\cortext[cor1]{Corresponding author}

\affiliation[1]{organization={Department of Mathematics and Computer Science, University of Ferrara},
addressline={Via Machiavelli 30},
city={Ferrara},
postcode={44121},
country={Italy}}




\begin{abstract}
We prove the holding of chaos in the sense of Li-Yorke for a family of four-dimensional discrete dynamical systems that are naturally associated to ODE systems describing coupled oscillators subject to an external non-conservative force, also giving an example of a discrete map that is Li-Yorke chaotic but not topologically transitive. Analytical results are generalized to a modular definition of the problem and to a system of nonlinear oscillators described by polynomial potentials in one coordinate. We perform numerical simulations looking for a strange attractor of the system; furthermore, we perform a bifurcation analysis of the system presenting 1D and 2D bifurcation diagrams, together with spectra of Lyapunov exponents and basins of attraction.
\end{abstract}




\end{frontmatter}



\section{Introduction}\label{sec_introduction}
The pioneering work \cite{Li} by Li and Yorke stated for the first time a formal definition of chaos for discrete dynamical systems; their main result, well summed up in the statement ``period three implies chaos'', gave a very simple sufficient condition in order to establish the arising of chaotic behavior for one-dimensional maps. Since the formulation of the Li-Yorke theorem, mathematicians put their efforts in order to formulate alternative definitions of discrete chaos, by exploring analytical, metric and topological aspects. In addition to the Li-Yorke chaos, it is worth to cite the Devaney chaos \cite{Devaney, Banks}, the Kolmogorov-Sinai entropy \cite{Kolmogorov1, Sinai}, the Block-Coppel chaos \cite{Block}, the distributional chaos \cite{Schweizer}. Some of these definitions may show connections one with the other, depending on the topological structure of the space; for example, in \cite{Aulbach} it is shown that, for continuous mappings of a compact interval into itself, Devaney chaos and Block-Coppel chaos are equivalent and imply Li-Yorke chaos, but not vice versa. This last definition of discrete chaos has been given originally for one-dimensional maps; however, it is easily formulated in higher dimensions.

It is worth to mention some recent results. In \cite{Bernardes} Li-Yorke chaos is analytically studied for linear operators on Banach and Fréchet spaces, while \cite{Wang} deeply focuses on a stronger definition of dense uniform Li-Yorke chaos. In \cite{Zhang1} it is shown how Li-Yorke chaos depends on the topology defining the dynamical system. The topic has been explored also for PDEs \cite{Zhang2}. See \cite{Demirovic, Wei} for explicit derivations of Li-Yorke chaos for specific models. An interesting emerging field concerns chaotification techniques for the so-called ``anti-control of chaos''; see \cite{Moysis} for an exhaustive review.

Aim of this work is to provide results for deeper studies about connections between discrete and continuous chaos; in particular, theorems we prove show how several continuous dynamical systems are chaotic in the sense of Li-Yorke if discretized in a proper way, regardless the integrability or chaoticity of the original system. This fact suggests that Li-Yorke chaos has a sort of preferential route from continuous systems to discrete ones.

The paper is organized as follows. In Section \ref{sec_prel} we recall some basic definitions together with Li-Yorke theorem and Marotto-Li-Chen theorem \cite{Marotto, Li-Chen1}; then, we define the discrete map subject of this work together with a simple notion of discretization. In Section \ref{sec_tran} we prove that the system is not topologically transitive for a symmetric choice on the parameters. In Section \ref{sec_theo} we prove the emergence of Li-Yorke chaos for proper choices on the parameters and generalize results for a modular definition of the system, leading to three chaotic discretization theorems for forced linear and nonlinear coupled oscillators. In Section \ref{sec_theo_comment} we comment a possible connection between Li-Yorke chaos and KAM theory \cite{Kolmogorov2, Arnold2, Moser} for quasi-integrable systems. In Section \ref{sec_numerical} we present some numerical results. In Section \ref{sec_numerical_strange} we perform qualitative studies of the orbit of the system, by comparing regular orbits for a symmetric choice on the parameters and the rising of a confined strange attractor in the general case. In Section \ref{sec_numerical_bifurcation} we present 1D and 2D bifurcation diagrams of the system and confirm the rising of period doubling bifurcations through a bifurcation analysis using the MATLAB tool \textsc{MatContM} (\href{https://matcont.sourceforge.io/}{https://matcont.sourceforge.io/}). Furthermore, we present spectra of Lyapunov exponents, that confirm sensitive dependence on initial data for several range of the parameters, together with basins of attraction of the system. Finally, in Section \ref{sec_conclusions} we state our conclusions and propose further developments of this work.

\section{Mathematical preliminaries}\label{sec_prel}

We begin this section by recalling the definition of chaos in the sense of Li-Yorke \cite{Li}, opportunely generalized for higher dimension dynamical systems. In the following we consider a discrete map $f \colon (X, d) \to (X, d)$ defined on a metric space $(X, d)$.
\begin{defi}\label{LY_defi}
A map $f$ is \emph{Li-Yorke chaotic} if there exists an uncountable set $S \subset X$ such that $\forall x, y \in S$, $x \ne y$:
\[
\liminf_{n \to +\infty} d( f^{(n)}(x) \,, f^{(n)}(y) ) = 0 \,,
\]
\[
\limsup_{n \to +\infty} d( f^{(n)}(x) \,, f^{(n)}(y) ) > 0 \,.
\]
\end{defi}
The following fundamental result is an immediate consequence of the main theorem presented in \cite{Li}.
\begin{theo}[Li-Yorke, 1975]\label{LY_theo}
Let $f \colon J \to J$ be continuous on an interval $J \subset \bb{R}$ such that there exists a periodic point of period 3; then, $f$ is Li-Yorke chaotic.
\end{theo}
We stress the fundamental consequence of the Li-Yorke theorem: any one-dimensional map defined on intervals that has a periodic point of period 3 will show itself as chaotic in the sense of Li-Yorke, i.e. it will admit a uncountable set of points that are simultaneously close and distant in the sense of Definition \ref{LY_defi}.

Li-Yorke theorem holds for one-dimensional continuous map from a real interval into itself, while in general a $3$-periodic point does not imply Li-Yorke chaos for $n$-dimensional maps. A sufficient condition for the onset of Li-Yorke chaos for $n$-dimensional discrete maps has been given in \cite{Marotto} through the existence of a snap-back repeller fixed point. Marotto theorem presented in \cite{Marotto} originally exhibited some problems in the proof; it has been well formulated and proved by Li-Chen in \cite{Li-Chen1} (see also \cite{Li-Chen2} for further discussions about the theorem).

Here we recall the definition of snap-back repeller and the Marotto-Li-Chen theorem, that we will use for the proof of theorems subject of this work.
\begin{defi}\label{snap_back_def}
It is given the $n$-dimensional discrete dynamical system
\[
f \colon (\bb{R}^n \,, ||\cdot||) \to (\bb{R}^n \,, ||\cdot||)\,, \quad x_{j+1} = f(j_n), \quad j \in \bb{N} \,.
\]
We denote with $J[f](x)$ the Jacobian of $f$ evaluated at a point $x \in \bb{R}^n$ and with $\lambda_i(x)$, $i = 1, \dots, n$ the eigenvalues of $(J^T[f] J[f])(x)$; finally, we denote with $B_r(p) \subset \bb{R}^n$ a closed ball centered in $p \in \bb{R}^n$ with radius $r > 0$. \\A fixed point $p_0$ is a \emph{snap-back repeller} if:
\begin{itemize}
\item $f$ is continuously differentiable in $B_r(p_0)$ for some $r > 0$;
\item $\lambda_i(p_0) > 1$ $\forall i = 1, \dots, n$;
\item $\exists q_0 \in B_{r'}(p_0) = \{ x \colon ||x-p_0|| \le r' \le r$, $\lambda_i(x) > 1$ $\forall i = 1, \dots, n \}$ such that:
\begin{itemize}
\item $q_0 \ne p_0$;
\item $\exists k \ge 1$ such that $f^{(k)}(q_0) = p_0$;
\item $f^{(i)}(q_0) \in B_r(p_0)$ $\forall i = 0, 1, \dots, k$;
\item $\det{ J[f^{(k)}](q_0) } \ne 0$.
\end{itemize}
\end{itemize}
\end{defi}
\begin{theo}[Marotto, 1978; Li-Chen, 2003]\label{MLC_theo}
If a $n$-dimensional discrete map has a snap-back repeller, then it is chaotic in the sense of Li-Yorke.
\end{theo}
In the vicinity of a snap-back repeller, there exists a set of initial conditions such that they exactly back onto the fixed point, after a certain amount of iterations. Geometrically, this behavior induces to a continuous stretching and folding of the phase space around the snap-back repeller, leading to emergence of chaos.

In this work we will use the following notion of discretization: given an autonomous continuous dynamical system defined by ODEs
\[\begin{split}
&f \colon \bb{R}^k \to \bb{R}^k \,, \\
&\frac{dx(t)}{dt} = f(x(t)) \,, \quad t \in \bb{R} \,,
\end{split}\]
we define its \textit{fixed point discretization} as the discrete map given by
\[\begin{split}
&f \colon \bb{R}^k \to \bb{R}^k \,, \\
&x_{n+1} = f(x_n) \,, \quad n \in \bb{N} \,.
\end{split}\]
Notice that the discretization notion we define is not a standard discretization algorithm for ODE systems, such as explicit or implicit Euler discretization (see, e.g., \cite{Kazakov}), and in general it is not appropriate for the numerical resolution of the associate ODE system. For instance, we can easily prove that for Hamiltonian systems it does not preserve in general the symplectic structure of the phase space, as shown in Proposition \ref{prop_symplectic}.
\begin{prop}\label{prop_symplectic}
The fixed point discretization is not symplectic for general Hamiltonian systems.
\end{prop}
\begin{proof}
We consider an Hamiltonian system defined by the Hamiltonian $H = H(q^{(j)}, p^{(j)})$ for $j = 1, \dots, m$ and the canonical coordinates $\overline{x} = (\overline{q}, \overline{p}) \in \bb{R}^{2m}$. The Hamilton equations are
\beq\begin{cases}
\dot q^{(j)} = \frac{\partial H}{\partial p^{(j)}} \\
\dot p^{(j)} = -\frac{\partial H}{\partial q^{(j)}} \,,
\end{cases}\eeq
therefore the fixed point discretization is given by
\beq\label{fp_hamiltonian}
\begin{cases}
q^{(j)}_{n+1} = \frac{\partial H}{\partial p^{(j)}_n} \\
p^{(j)}_{n+1} = -\frac{\partial H}{\partial q^{(j)}_n} \,.
\end{cases}
\eeq
It is sufficient whether the symplecticity condition is satisfied, that is
\beq
J^T M J = M \,, \quad M =
\begin{pmatrix}
0 & \bb{1}_{m \times m} \\
-\bb{1}_{m \times m} & 0
\end{pmatrix} \,,
\eeq
where J is the Jacobian of the system \eqref{fp_hamiltonian}. A simple calculation is sufficient to see that the symplecticity condition holds only for the specific class of Hamiltonians that satisfy
\beq
C^T A = (C^T A)^T \,, \quad A^T B = (A^T B)^T \,, \quad C^T B - A^T A = \bb{1}_{m \times m} \,,
\eeq
where
\beq\begin{split}
&A_{jk} = \frac{\partial}{\partial q^{(j)}_n} \bigg( \frac{\partial H}{\partial p^{(k)}_n} \bigg) \,, \\
&B_{jk} = \frac{\partial}{\partial p^{(j)}_n} \bigg( \frac{\partial H}{\partial p^{(k)}_n} \bigg) \,, \\
&C_{jk} = \frac{\partial}{\partial q^{(j)}_n} \bigg( \frac{\partial H}{\partial q^{(k)}_n} \bigg) \,.
\end{split}\eeq
\end{proof}
As an example, for one-dimensional Hamiltonian $H = H(q, p)$, $q, p \in \bb{R}$, the fixed point discretization is symplectic only for Hamiltonian with unitary Hessian.

Nevertheless, our aim is not to obtain a discrete approximation of the ODE system, but rather to define an alternative discrete version associated to the continuous system and to explore connections between regular and chaotic behavior of both systems. The discrete system can share some properties with the continuous one (e.g. fixed points); on the other hand, it can exhibit a richer dynamics. Furthermore, the discretization we define can be thought as a novel chaotification procedure that associates a Li-Yorke chaotic discrete system to an integrable continuous one, under suitable constraints on the parameters.

Goal of this paper is to show that a large class of continuous dynamical systems has a fixed point discretization that is Li-Yorke chaotic for non trivial choices on the parameters, regardless the integrability or chaoticity of the continuous system. The starting point will be the discrete map defined in \cite{Disca1, Disca2}; in this work we relax the original constraints on the coefficients and define the \textit{Ziegler discrete map} as
\beq\label{Ziegler_map}
\begin{split}
&\qquad f \colon \bb{R}^4 \to \bb{R}^4 \,, \\
&\begin{cases}
x_{n+1} = y_n \\
y_{n+1} = \alpha x_n + \beta z_n + \gamma \sin(x_n) \\
z_{n+1} = \omega_n \\
\omega_{n+1} = -y_{n+1} - \sigma x_n \,,
\end{cases} \quad \alpha \,, \beta \,, \gamma \,, \sigma \in \bb{R} \,.
\end{split}
\eeq
The map \eqref{Ziegler_map} is obtained by discretizing in the above sense the equations of motion of the Ziegler pendulum in the sense of Polekhin \cite{Polekhin} for a choice on the parameters leading to non-Hamiltonian integrability. From a physical point of view, parameters $\alpha$ and $\sigma$ involve masses and lengths of the pendulum, $\beta$ models an angular elastic potential and $\gamma$ represents the modulus of an external non-conservative force acting on the system. In the following we ignore this physical interpretation, focusing on general ranges on the parameters and emerging of chaotic behavior.

We also define the \textit{modular Ziegler discrete map} as
\beq\label{Ziegler_map_mod}
\begin{split}
&\qquad \tilde f \colon \bb{R}^4 \to [0, 2\pi) \times \bb{R} \times [0, 2\pi) \times \bb{R} \,, \\
&\begin{cases}
x_{n+1} = y_n \quad \mod{2\pi} \\
y_{n+1} = \alpha x_n + \beta z_n + \gamma \sin(x_n) \\
z_{n+1} = \omega_n \quad \mod{2\pi} \\
\omega_{n+1} = -y_{n+1} - \sigma x_n \,.
\end{cases} \quad \alpha \,, \beta \,, \gamma \,, \sigma \in \bb{R} \,.
\end{split}
\eeq

\section{Topological transitivity for the Ziegler map}\label{sec_tran}
We recall the definition of chaos in the sense of Devaney \cite{Devaney}, that is formulated through topological transitivity and density of periodic points.
\begin{defi}
A discrete map $f \colon \bb{R}^k \to \bb{R}^k$ is \emph{topologically transitive} if $\forall V, W \subset \bb{R}^k$ open sets $\exists n \ge 1$ such that $V \cap f^{(n)}(W) \ne \emptyset$.
\end{defi}
\begin{defi}
A discrete map $f$ is \emph{Devaney chaotic} if it is topologically transitive and there exists a dense set of periodic points.
\end{defi}
\begin{rem}
It is worth to notice that the original definition of Devaney was also considering sensitive dependence on initial data, a typical signature of chaotic systems; this condition has been proved to be redundant in \cite{Banks}.
\end{rem}
In this section we reintroduce the Ziegler map as defined in \cite{Disca1, Disca2}, that is
\beq\label{Ziegler_map_original}
\begin{split}
&\qquad f \colon \bb{R}^4 \to \bb{R}^4 \,, \\
&\begin{cases}
x_{n+1} = y_n \\
y_{n+1} = \alpha x_n + \beta z_n + \gamma \sin(x_n) \\
z_{n+1} = \omega_n \\
\omega_{n+1} = \tilde \alpha x_n - \beta z_n - \gamma \sin(x_n) \\
\end{cases} \quad \alpha \,, \tilde \alpha \,, \beta \,, \gamma \in \bb{R} \,.
\end{split}
\eeq
In \cite{Disca2} it is proved that the map \eqref{Ziegler_map_original} is not Devaney chaotic for $\alpha < 0$, $\tilde \alpha > 0$, $\beta > 0$, $\gamma \ne 0$, since it has not dense sets of periodic points; in particular, the map has not even periodic points and the sets of odd periodic points are not dense in $\bb{R}^4$. We show that the remaining necessary condition for the Devaney chaos is not satisfied for a symmetric choice on the parameters.
\begin{prop}
The map \eqref{Ziegler_map_original} is not topologically transitive for $\tilde \alpha = -\alpha$, $\beta \ne 0$.
\end{prop}
\begin{proof}
Let be defined the Ziegler map \eqref{Ziegler_map_original} on the space $(\bb{R}^4, d)$, where $d$ is the distance induced by a certain norm on $\bb{R}^4$. We fix the notation $W_{n+1} := f^{(n+1)}(W)$ and denote with $\overline{W}_{n+1}$ the closure of $W_{n+1}$. Finally, we set $d(A, B) = \inf_{x \in A, y \in B} d(a, b)$ $\forall A, B \subset \bb{R}^4$. We have
\beq
V \cap \overline{W}_{n+1} = \emptyset \iff d(V \,, \overline{W}_{n+1}) > 0 \quad \iff d(p, q) > 0 \quad \forall p \in V \,, q \in \overline{W}_{n+1}
\eeq
 $\forall V, W \in \bb{R}^4$. We set $f^{(n+1)}(p) = (p_x, p_y, p_z, p_\omega)$ with $p \in \bb{R}^4$; then, $\forall p, q \in \bb{R}^4$
\beq\begin{split}
f^{(n+1)}(p) - f^{(n+1)}(q) = \bigg(& p_y - q_y \,, \alpha (p_x - q_x) + \beta (p_z - q_z) + \gamma ( \sin(p_x) - \sin(q_x) ) \,, \\
&p_\omega - q_\omega \,, \tilde \alpha (p_x - q_x) - \beta (p_z - q_z) - \gamma ( \sin(p_x) - \sin(q_x) ) \bigg) \,.
\end{split}\eeq
Therefore, $\forall p \in \bb{R}^4$ and $n \in \bb{N}$ we choose the point $q \in \bb{R}^4$ such that
\beq
p_x = q_x + 2 \pi \,, \quad p_y = q_y \,, \quad p_z = q_z \,, \quad p_\omega = q_\omega \,,
\eeq
so that
\beq
f^{(n+1)}(p) - f^{(n+1)}(q) = \bigg( 0 \,, 2\pi \alpha + \beta (p_z - q_z) \,, 0 \,, 2\pi \tilde \alpha - \beta (p_z - q_z) \bigg) \,.
\eeq
By imposing $\tilde \alpha = -\alpha$, the point $q$ is such that $f^{(n)}(q)= (p_x - 2\pi, p_y, p_y + \frac{2\pi \alpha}{\beta}, p_\omega )$; in particular, $\forall p \in \bb{R}^4$ we can define a sequence $\{ q(p)_k \}_{k \in \bb{Z}} \subset \bb{R}^4$ such that
\beq
f^{(n)}( q(p)_k )= \bigg( p_x - 2 k \pi \,, p_y \,, p_y + \frac{2 k \pi \alpha}{\beta} \,, p_\omega \bigg) \implies f^{(n+1)}(p) - f^{(n+1)}(q) = 0 \,.
\eeq
For $n \in \bb{N}$ and $p \in V$, $q \in W_{n+1}$, $d(p, q) > 0 \iff q \not\in \{ q(p)_k \}_{k \in \bb{Z}}$; hence, we define
\beq
\tilde {W}_p = W \setminus \bigcup_{n \in \bb{N} \,, k \in \bb{Z}} \{ f^{(-n)}( q(p)_k ) \}
\eeq
and finally
\beq
\tilde W = \Int \bigg( \bigcup_{p \in V} \tilde W_p \bigg) \,,
\eeq
where $\Int(V)$ is the interior of the set $V$. By definition, $V$ and $\tilde W$ are open sets; by construction, $d( V, f^{(n+1)}(\tilde W) ) > 0$ $\forall n \in \bb{N}$ $\implies$ $V \cap \overline{ f^{(n+1)}(\tilde W) } = \emptyset$ $\forall n \in \bb{N}$ $\implies$ $V \cap f^{(n+1)}(\tilde W) = \emptyset$ $\forall n \in \bb{N}$.
\end{proof}

\section{Chaotic discretization theorems}\label{sec_theo}
In this section we focus on Li-Yorke chaos for the maps \eqref{Ziegler_map} and \eqref{Ziegler_map_mod}; in particular, we will show that perturbative formulations of the maps admit continuous set of values for the parameters such that systems are Li-Yorke chaotic for arbitrary small perturbations.

Before going into the problem, we derive the following simple result.
\begin{prop}\label{prop_link}
If $(x_0, y_0, z_0, \omega_0) \in (0, 2\pi)^4$ and
\beq\begin{cases}\label{constraint_link}
\alpha > 0, \beta > 0, \gamma \in (-2\pi, 0) \\
\sigma < -1 \\
\alpha + \beta < 1 + \frac{\gamma}{2\pi} \,,
\end{cases}\eeq
the map \eqref{Ziegler_map} is Li-Yorke chaotic if and only if the map \eqref{Ziegler_map_mod} is Li-Yorke chaotic.
\end{prop}
\begin{proof}
Given the map \eqref{Ziegler_map}, we prove that constraints \eqref{constraint_link} are sufficient for $(x_n, y_n, z_n, \omega_n)$ to be confined in $(0, 2\pi)^4$ $\forall n \ge 2$. \\From the second equation in \eqref{Ziegler_map} we get
\beq
y_{n+2} = \alpha x_{n+1} + \beta z_{n+1} + \gamma \sin(x_{n+1}) = \alpha \tilde y_n + \beta \tilde \omega_n + \gamma \sin(\tilde y_n) \,.
\eeq
It follows
\beq
y_{n+2} \le |\alpha| |\tilde y_n| + |\beta| |\tilde \omega_n| + |\gamma| \le 2\pi (|\alpha| + |\beta|) + |\gamma| \,,
\eeq
therefore
\beq
|\alpha| + |\beta| < 1 - \frac{|\gamma|}{2\pi} \implies y_{n} < 2\pi \quad \forall n \ge 2
\eeq
(notice that it must be $|\gamma| < 2\pi$). On the other hand, if $\alpha > 0$, $\beta > 0$, we have
\beq
y_{n+2} = \alpha \tilde y_n + \beta \tilde \omega_n + \gamma \sin(\tilde y_n) \ge -\gamma \,,
\eeq
therefore
\beq
\alpha > 0, \beta > 0, \gamma < 0 \implies y_{n} > 0 \quad \forall n \ge 2 \,.
\eeq
Analogously, from the fourth equation in \eqref{Ziegler_map} we get
\beq
\omega_{n+2} = -y_{n+2} - \sigma x_{n+1} = -y_{n+2} - \sigma \tilde y_n \,.
\eeq
It follows
\beq
\omega_{n+2} < 2\pi \iff -y_{n+2} < 2\pi + \sigma \tilde y_n < 2\pi (1 + \sigma) < 2\pi \iff \sigma < -1 \,,
\eeq
therefore
\beq
\sigma < -1 \implies \omega_{n} < 2\pi \quad \forall n \ge 2 \,.
\eeq
On the other hand, we have
\beq
\omega_{n+2} > -2\pi(1 + \sigma) > 0 \iff \sigma < -1 \,,
\eeq
therefore
\beq
\sigma < -1 \implies \omega_{n} > 0 \quad \forall n \ge 2 \,.
\eeq
If $(x_0, y_0, z_0, \omega_0) \in (0, 2\pi)^4$, the previous reasoning can be applied also to $y_1, \omega_{1}$. Finally, $x_{n+1}$, $z_{n+1} \in [0, 2\pi)$ $\forall n \ge 1$ by definition; since again $(x_0, y_0, z_0, \omega_0) \in (0, 2\pi)^4$, $x_n \in (0, 2\pi)$, $z_n \in (0, 2\pi)$ $\forall n \ge 1$. \\Finally we conclude that, if \eqref{constraint_link} holds, then
\beq
f \bigg|_{ (x_0, y_0, z_0, \omega_0) \in (0, 2\pi)^4 } = \tilde f \,.
\eeq
\end{proof}
Now, we want to derive proper constraints on parameters and initial conditions such that the map \eqref{Ziegler_map} satisfies Theorem \ref{MLC_theo}. Unfortunately, \eqref{Ziegler_map} does not satisfy one of conditions stated in Theorem \ref{MLC_theo} in any case; indeed, by setting $A = A(x) := \alpha + \gamma \cos(x)$, the Jacobian of \eqref{Ziegler_map} is
\beq
J[f](x) =
\begin{pmatrix}
0 & 1 & 0 & 0 \\
A & 0 & \beta & 0 \\
0 & 0 & 0 & 1 \\
-A - \sigma & 0 & -\beta & 0
\end{pmatrix} \,,
\eeq
so we have
\beq
(J^T J)[f](x) =
\begin{pmatrix}
A^2 + (A+\sigma)^2 & 0 & \beta(2A + \sigma) & 0 \\
0 & 1 & 0 & 0 \\
A^2 + (A+\sigma)^2 & 0 & 2 \beta^2 & 0 \\
0 & 0 & 0 & 1
\end{pmatrix} \,.
\eeq
Therefore, $\lambda = 1$ is eigenvalue of $J^T J[f](x)$ $\forall x \in \bb{R}^4$.

In order to overcome this problem, we define the following perturbed formulation of \eqref{Ziegler_map}:
\beq\label{Ziegler_map_pert}
f \colon \bb{R}^4 \to \bb{R}^4 \,, \quad
\begin{cases}
x_{n+1} = (1 + \varepsilon_1) y_n + \varepsilon_2 x_n \\
y_{n+1} = \alpha x_n + \beta z_n + \gamma \sin(x_n) \\
z_{n+1} = (1 + \varepsilon_1) \omega_n \\
\omega_{n+1} = -y_{n+1} - \sigma x_n \,,
\end{cases}
\eeq
and derive our results for the map \eqref{Ziegler_map_pert} with arbitrary small values of $\varepsilon_1$, $\varepsilon_2$; naturally, \eqref{Ziegler_map_pert} reduces to \eqref{Ziegler_map} for $\varepsilon_1, \varepsilon_2 \to 0^+$. From a physical point of view, the coefficients $\varepsilon_1$, $\varepsilon_2$ model a (small) negative-dumping mechanism acting on the system.

We divide the proof in a few steps, by analyzing separately the sufficient conditions stated in Theorem \ref{MLC_theo} and deriving suitable constraints on the parameters.
\begin{lem}[continuous differentiability]\label{lem_diff}
The maps \eqref{Ziegler_map_pert} is continuously differentiable.
\end{lem}
\begin{proof}
The proof is immediate, since \eqref{Ziegler_map_pert} is sum of continuously differentiable functions.
\end{proof}
\begin{lem}[existence of fixed points]\label{lem_fixed}
There exist non-vanishing fixed points for the map \eqref{Ziegler_map_pert}, provided that
\beq\label{constraint_fixed}
\inf_{x \in \bb{R} }{ \frac{\sin(x)}{x} } \le \frac{ 1 - \alpha - \beta (1 + \sigma) }{\gamma} \le 1
\eeq
and $\varepsilon_1, \varepsilon_2$ sufficiently small.
\end{lem}
\begin{proof}
Fixed points of \eqref{Ziegler_map} are known from \cite{Disca1}; by considering the new definitions given in \eqref{Ziegler_map} and \eqref{Ziegler_map_mod}, we can write them through the solutions of
\beq\label{eq_fixed}
\sin(y_0) = \frac{ 1 - \alpha - \beta (1 + \sigma) }{\gamma} y_0 \,,
\eeq
so if $y_0$ is a solution of \eqref{eq_fixed}, then $(x_0, y_0, z_0, \omega_0) = (y_0, y_0, -(1 + \sigma) y_0, -(1+\sigma) y_0 )$ is a fixed point of \eqref{Ziegler_map}. By imposing \eqref{constraint_fixed}, the equation \eqref{eq_fixed} always admits solutions in $\bb{R}$: only $y_0 = 0$ if the right equality holds, an odd number of symmetric solutions that increases as $1 - \alpha - \beta(1+\sigma)$ approaches to zero, all the integer multiples of $\pi$ if the term is equal to zero, two symmetric solutions if the left equality holds. For continuity, result holds for \eqref{Ziegler_map_pert} and $\varepsilon_1$, $\varepsilon_2$ sufficiently small.
\end{proof}
\begin{lem}[eigenvalues of $J^T J$ greater than $1$]\label{lem_jacobian}
Let $J$ be the Jacobian of \eqref{Ziegler_map_pert}. $\forall \varepsilon_1, \varepsilon_2 > 0$, there exist $\overline \sigma (\varepsilon_1, \varepsilon_2) \in \bb{R}$ and a fixed point $p_0$ such that the matrix $J^T J(p_0)$ has all eigenvalues greater than $1$, provided that $|\beta| > \frac{1}{\sqrt{2}}$ and $|\sigma| > \overline \sigma$.
\end{lem}
\begin{proof}
By setting $A = A(x) := \alpha + \gamma \cos(x)$, the Jacobian of \eqref{Ziegler_map_pert} is
\beq
J[f](x) =
\begin{pmatrix}
\varepsilon_2 & 1+\varepsilon_1 & 0 & 0 \\
A & 0 & \beta & 0 \\
0 & 0 & 0 & 1+\varepsilon_1 \\
-A - \sigma & 0 & -\beta & 0
\end{pmatrix} \,,
\eeq
so we have
\beq
(J^T J)[f](x) =
\begin{pmatrix}
\varepsilon_2^2 + A^2 + (A+\sigma)^2 & \varepsilon_2(1+\varepsilon_1) & \beta(2A + \sigma) & 0 \\
\varepsilon_2(1+\varepsilon_1) & (1+\varepsilon_1)^2 & 0 & 0 \\
\beta(2A + \sigma) & 0 & 2 \beta^2 & 0 \\
0 & 0 & 0 & (1+\varepsilon_1)^2
\end{pmatrix} \,.
\eeq
Now, given the fixed points stated in Lemma \ref{lem_fixed}, we choose $x_0 = y_0$ and parameters such that $2 A(y_0) + \sigma = 0$; this is equivalent to solve the system
\beq\begin{cases}\label{eq_tan_syst}
\sin(y_0) = \frac{ 1 - \alpha - \beta (1 + \sigma) }{\gamma} y_0 \\
\cos(y_0) = -\frac{\sigma + 2 \alpha}{2 \gamma} \,,
\end{cases}\eeq
i.e.
\beq\label{eq_tan}
\tan(y_0) = 2 \frac{ \alpha + \beta(1+\sigma) - 1 }{\sigma + 2 \alpha} y_0 \,.
\eeq
Equation \eqref{eq_tan} admits always solution; for the singular cases $\sigma = -2\alpha$, $y_0 = \frac{\pi}{2} + 2k\pi$, $k \in \bb{Z}$ return to the system \eqref{eq_tan_syst}. Therefore, we have
\beq
(J^T J)[f](y_0) =
\begin{pmatrix}
\varepsilon_2^2 \sigma^2 & \varepsilon_2(1+\varepsilon_1) & 0 & 0 \\
\varepsilon_2(1+\varepsilon_1) & (1+\varepsilon_1)^2 & 0 & 0 \\
0 & 0 & 2 \beta^2 & 0 \\
0 & 0 & 0 & (1+\varepsilon_1)^2
\end{pmatrix} \,.
\eeq
The eigenvalues of $(J^T J)[f](y_0)$ are
\beq\begin{split}
&\lambda_1 = (1+\varepsilon_1)^2 \,, \\
&\lambda_\beta = 2 \beta^2 \,, \\
&\lambda_\pm = \frac{ (1 + \varepsilon_1)^2 + \varepsilon_2^2 + \frac{\sigma^2}{2} }{2} \pm \frac{1}{2} \sqrt{ \bigg( (1 + \varepsilon_1)^2 + \varepsilon_2^2 + \frac{\sigma^2}{2} \bigg)^2 - 2 \sigma^2 (1 + \varepsilon_1)^2 } \,.
\end{split}\eeq
We have $\lambda_{\pm} > 1$ for a sufficiently great value for the modulus of $\sigma$; furthermore, $\lambda_\beta > 1$ if $|\beta| > \frac{1}{\sqrt{2}}$. Finally, $\lambda_1 > 1$ $\forall \varepsilon_1 > 0$. We conclude that, $\forall \varepsilon_{1, 2} > 0$ and $|\beta| > \frac{1}{\sqrt{2}}$, there exists a threshold parameter $\overline{\sigma}(\varepsilon_1, \varepsilon_2) \in \bb{R}$ such that $\lambda_1, \lambda_\beta, \lambda_\pm > 1$, provided that $|\sigma| > \overline \sigma$.
\end{proof}
\begin{lem}[existence of a snap-back repeller]\label{lem_snap_back}
There exists a snap-back repeller for the map \eqref{Ziegler_map_pert}, provided that $\beta \ne 0$, $\sigma \ne 0$ and $\varepsilon_1$, $\varepsilon_2$ sufficiently small.
\end{lem}
\begin{proof}
Given Lemma \ref{lem_jacobian}, we can avoid the need to construct a proper neighborhood of $p_0$. Furthermore, the map \eqref{Ziegler_map} has not even periodic points \cite{Disca2}, so we can show the existence of a snap-back for the point $p_0$ by proving the existence of solutions for the equation $f^{(2)}(q) = p$; since it cannot exists a periodic point of period 2, it follows that $q_0 \ne p_0$. \\The second iterate of \eqref{Ziegler_map} is
\beq\begin{cases}
x_{n+2} = \alpha x_n + \beta z_n + \gamma \sin(x_n) \\
y_{n+2} = \alpha y_n + \beta \omega_n + \gamma \sin(y_n) \\
z_{n+2} = -(\alpha + \sigma) x_n - \beta z_n - \gamma \sin(x_n) \\
\omega_{n+2} = -(\alpha + \sigma) y_n - \beta \omega_n - \gamma \sin(y_n) \,.
\end{cases}\eeq
By naming $p_0 = (x_0, y_0, z_0, \omega_0)$, $q_0 = (\overline{x}_0, \overline{y}_0, \overline{z}_0, \overline{\omega}_0 )$ and defining the function $g(x, z) = \alpha x + \beta z + \gamma \sin(x)$, the equation $f^{(2)}(q) = p$ reads
\beq\begin{cases}\label{eq_iterate2}
g(\overline{x}_0, \overline{z}_0) = x_0 \\
g(\overline{y}_0, \overline{\omega}_0) = y_0 \\
- g(\overline{x}_0, \overline{z}_0) - \sigma \overline{x}_0 = z_0 \\
- g(\overline{y}_0, \overline{\omega}_0) - \sigma \overline{y}_0 = \omega_0 \,.
\end{cases}\eeq
From Lemma \ref{lem_fixed}, fixed points must be of the form $(y_0, y_0, \omega_0, \omega_0 )$, so also $q_0$ must be of the form $(\overline{y}_0, \overline{y}_0, \overline{\omega}_0, \overline{\omega}_0)$. Therefore, \eqref{eq_iterate2} becomes
\beq\begin{cases}\label{eq_iterate2}
g(\overline{y}_0, \overline{\omega}_0) = y_0 \\
-g(\overline{y}_0, \overline{\omega}_0) -\sigma \overline{y}_0 = \omega_0 \,,
\end{cases}\eeq
i.e. $\overline{y}_0 = - \frac{y_0 + \omega_0}{\sigma}$ and $\overline{\omega}_0$ solution of $g \bigg( - \frac{y_0 + \omega_0}{\sigma}, \overline{\omega}_0 \bigg) = y_0$. Finally, we verify that $f^{(2)}$ is a diffeomorphism in $q_0$; by following the notation of Lemma \ref{lem_fixed}, we have
\beq
\det{ J[f^{(2)}](x) } = \det{
\begin{pmatrix}
A(x) & 0 & \beta & 0 \\
0 & A(y) & 0 & \beta \\
-A(x) - \sigma & 0 & -\beta & 0 \\
0 & -A(y) - \sigma & 0 & -\beta
\end{pmatrix} } = \beta^2 \sigma^2 \,.
\eeq
Therefore, $\det{ J[f^{(2)}](x) } \ne 0$ $\forall x \in \bb{R}^4$, provided that $\beta \ne 0$, $\sigma \ne 0$. For continuity, result holds for \eqref{Ziegler_map_pert} and $\varepsilon_1$, $\varepsilon_2$ sufficiently small.
\end{proof}
We collect Lemma \ref{lem_diff} - \ref{lem_snap_back} into the following final results.
\begin{theo}\label{theo_CD1}
It is given the family of discrete dynamical systems $f \colon \bb{R}^4 \to \bb{R}^4$ defined by
\beq\label{map_theo_CD1}
\begin{cases}
x_{n+1} = (1 + \varepsilon_1) y_n + \varepsilon_2 x_n \\
y_{n+1} = \alpha x_n + \beta z_n + \gamma \sin(x_n) \\
z_{n+1} = (1 + \varepsilon_1) \omega_n \\
\omega_{n+1} = -y_{n+1} - \sigma x_n \,,
\end{cases}
\eeq
where $\varepsilon_{1,2} > 0$, $\alpha \in \bb{R}$, $\gamma \in \bb{R}$, $\sigma \ne 0$ and
\beq
|\beta| > \frac{1}{\sqrt{2}} \,.
\eeq
Then, for $\varepsilon_{1,2}$ sufficiently small, $\exists \overline{\sigma}(\varepsilon_1, \varepsilon_2) \in \bb{R}$ such that \eqref{map_theo_CD1} is chaotic in the sense of Li-Yorke $\forall \sigma$ such that $|\sigma| > \overline \sigma$.
\end{theo}
Notice that for Theorem \ref{theo_CD1} it is not necessary to consider Lemma \ref{lem_fixed}, since $(0, 0, 0, 0)$ is always fixed point of \eqref{map_theo_CD1}; instead, Lemma \ref{lem_fixed} is necessary for the following one, together with Proposition \ref{prop_link}.
\begin{theo}\label{theo_CD2}
It is given the family of discrete dynamical systems $f \colon (0, 2\pi)^4 \to (0, 2\pi)^4$ defined by
\beq\label{map_theo_CD2}
\begin{cases}
x_{n+1} = (1 + \varepsilon_1) y_n + \varepsilon_2 x_n \quad \mod 2\pi \\
y_{n+1} = \alpha x_n + \beta z_n + \gamma \sin(x_n) \\
z_{n+1} = (1 + \varepsilon_1) \omega_n \quad \mod 2\pi \\
\omega_{n+1} = -y_{n+1} - \sigma x_n \,,
\end{cases}
\eeq
where $\varepsilon_{1,2} > 0$, $\alpha > 0$, $\beta > 0$, $\gamma \in (-2\pi, 0)$, $\sigma < -1$ and
\beq\begin{cases}
\alpha + \beta < 1 + \frac{\gamma}{2\pi} \\
\inf_{x \in \bb{R} }{ \frac{\sin(x)}{x} }  \le \frac{ 1 - \alpha - \beta (1 + \sigma) }{\gamma} \le 1 \,.
\end{cases}\eeq
Then, for $\varepsilon_{1,2}$ sufficiently small, $\exists \overline{\sigma}(\varepsilon_1, \varepsilon_2) \in \bb{R}$ such that \eqref{map_theo_CD2} is chaotic in the sense of Li-Yorke $\forall \sigma$ such that $|\sigma| > \overline \sigma$.
\end{theo}
We have the following result for \eqref{Ziegler_map_pert} when generalized to arbitrary polynomial functions of $z_n$. Again, Lemma \ref{lem_fixed} is not necessary in this case.
\begin{theo}\label{theo_CD3}
It is given the family of discrete dynamical systems $f \colon \bb{R}^4 \to \bb{R}^4$ defined by
\beq\label{map_theo_CD3}
\begin{cases}
x_{n+1} = (1 + \varepsilon_1) y_n + \varepsilon_2 x_n \\
y_{n+1} = \alpha x_n + B^{(j)}(z_n) + \gamma \sin(x_n) \\
z_{n+1} = (1 + \varepsilon_1) \omega_n \\
\omega_{n+1} = -y_{n+1} + 2 \alpha x_n \,,
\end{cases}
\eeq
where $\varepsilon_{1,2} > 0$, $\alpha \ne 0$, $\gamma \in \bb{R}$,
\beq
B^{(j)}(z_n) = \sum_{k=1}^j \beta_k z_n^k \,, \quad |\beta_1| > \frac{1}{\sqrt{2}} \,, \quad \beta_k \in \bb{R} \quad \forall k \ge 2
\eeq
and $B^{(j)}$ has at least two distinct real roots. Then, for $\varepsilon_{1,2}$ sufficiently small, $\exists \overline{\alpha}(\varepsilon_1, \varepsilon_2) \in \bb{R}$ such that \eqref{map_theo_CD3} is chaotic in the sense of Li-Yorke $\forall \alpha$ such that $|\alpha| > \overline \alpha$.
\end{theo}
\begin{proof}
The map \eqref{map_theo_CD3} is continuously differentiable and $(0, 0, 0, 0)$ is always fixed point. Given
\beq
\tilde B(z_n) = \frac{d}{dz_n} B^{(j)}(z_n) = \sum_{k=1}^j k \beta_k z_n^{k-1} \,,
\eeq
it is sufficient to take the following substitutions in proofs of Lemma \ref{lem_jacobian} - \ref{lem_snap_back}:
\begin{itemize}
\item $y_0 = 0$;
\item $\beta \to \tilde B(0) = \beta_1$;
\item $\sigma = -2\alpha$;
\item $\lambda_\beta \to \lambda_{\beta_1} = 2 \beta_1^2$.
\end{itemize}
Finally, since $\omega_0 = 0$, equation \eqref{eq_iterate2} in Lemma \ref{lem_snap_back} becomes
\beq\begin{cases}
g(\overline{y}_0, \overline{\omega}_0) = 0 \\
-g(\overline{y}_0, \overline{\omega}_0) -\sigma \overline{y}_0 = 0 \,.
\end{cases}\eeq
so $\overline{y}_0 = \overline{x}_0 = 0$ and $\overline{\omega}_0 = \overline{z}_0$ are such that
\beq
g(0, \overline{\omega}_0) = 0 \implies B^{(j)}(\overline{\omega}_0) = 0 \,.
\eeq
Therefore, snap-back point is identified by the real roots of $B^{(j)}$; since $B^{(j)}$ has at least two distinct real roots, there exists a point $q_0 \ne (0, 0, 0, 0)$ such that $f^{(2)}(q_0) = (0, 0, 0, 0)$.
\end{proof}
A few comments are needed.

Firstly, even if the Marotto-Li-Chen theorem cannot be applied to the original system \eqref{Ziegler_map}, it has been successfully applied to the perturbed formulation \eqref{Ziegler_map_pert} for arbitrary small $\varepsilon_1$, $\varepsilon_2$; from a numerical point of view, sufficiently small values for $\varepsilon_1$, $\varepsilon_2$ lead to practically identical dynamics. However, as $\varepsilon_1$, $\varepsilon_2$ decrease, the threshold parameter $\overline{\sigma}$ increases, so that the range of parameters $\sigma$ associated to Li-Yorke chaos becomes smaller, up to the empty set for $\varepsilon_1 = \varepsilon_2 = 0$.

Secondly, Theorem \ref{theo_CD1} - \ref{theo_CD2} provide sufficient conditions for Li-Yorke chaos related to a family of forced linear coupled oscillators that are the fixed point discretization of the ODEs system
\beq
\begin{cases}
\ddot x(t) = (1 + \varepsilon_1) \bigg( \alpha x(t) + \beta z(t) + \gamma \sin(x(t)) \bigg) + \varepsilon_2 \dot x(t) \\
\ddot z(t) = -(1 + \varepsilon_1) \bigg( (\alpha + \sigma) x(t) + \beta z(t) + \gamma \sin(x(t)) \bigg) \\
x(0) = x_0 \,, \quad \dot x(0) = y_0 \,, \quad z(0) = z_0 \,, \quad \dot z(0) = \omega_0 \,,
\end{cases}
\eeq
\beq
\begin{cases}
\ddot x(t) = (1 + \varepsilon_1) \bigg( \alpha x(t) + \beta z(t) + \gamma \sin(x(t)) \bigg) + \varepsilon_2 \dot x(t) \quad \mod 2\pi \\
\ddot z(t) = -(1 + \varepsilon_1) \bigg( (\alpha + \sigma) x(t) + \beta z(t) + \gamma \sin(x(t)) \bigg) \quad \mod 2\pi \\
x(0) = x_0 \,, \quad \dot x(0) = y_0 \,, \quad z(0) = z_0 \,, \quad \dot z(0) = \omega_0 \,,
\end{cases}
\eeq
respectively. Analogously, Theorem \ref{theo_CD3} provide sufficient conditions for Li-Yorke chaos related to a family of forced nonlinear coupled oscillators that are the fixed point discretization of the ODEs system
\beq
\begin{cases}
\ddot x(t) = (1 + \varepsilon_1) \bigg( \alpha x(t) + \sum_{k=1}^j \beta_k z(t)^k + \gamma \sin(x(t)) \bigg) + \varepsilon_2 \dot x(t) \\
\ddot z(t) = -(1 + \varepsilon_1) \bigg( -\alpha x(t) + \sum_{k=1}^j \beta_k z(t)^k + \gamma \sin(x(t)) \bigg) \\
x(0) = x_0 \,, \quad \dot x(0) = y_0 \,, \quad z(0) = z_0 \,, \quad \dot z(0) = \omega_0 \,.
\end{cases}
\eeq
Theorem \ref{theo_CD3} is particularly relevant in the context of nonlinear oscillators \cite{Kovacic}, when they are perturbed by periodic external forces, defined by polynomial potentials and are subject to linear velocity-dependent perturbations.

As an explicit example of application, Theorem 3 - 5 can be applied to proper functional extensions of discrete models for cardiac arrhythmia, such as the Guevara map proposed in \cite{Guevara}.

\subsection{Qualitative remarks on scrambled sets and perturbations}\label{sec_theo_comment}
In this section we outline a heuristic argument regarding potential formal connections between Li-Yorke chaos and KAM theory for quasi-integrable dynamical systems. The following discussion is meant to be a roadmap for future research works in the field.

In our first attempt to prove Theorem \ref{theo_CD1}, we started from the definition of Li-Yorke chaos, namely the existence of a scrambled set for the discrete map, by showing that the system is Li-Yorke chaotic for the case $\sigma = 0$. Then, we wrote the system in the general case as the map corresponding to this choice, subject to a perturbation. Given this formulation of the problem, we assumed that the perturbed map shares a scrambled set with the unperturbed one.

Regardless of the actual flaws in the initial proof, a natural question arises: by assuming that a discrete map is Li-Yorke chaotic, are there sufficient conditions under which the perturbed system shares a scrambled set with the unperturbed one? In other words, how does a perturbation act on the set of chaotic initial conditions of a Li-Yorke chaotic map?

In answering this question, a curious connection with the KAM theorem \cite{Kolmogorov2, Arnold2, Moser} can be conjectured. From KAM theory, we know that a quasi-integrable system exhibits invariant tori in the phase space that are not destroyed by the presence of a perturbation; when acting on the system, the perturbation destroys some tori, while others survive but are deformed, according to specific non-resonance conditions. In the context of chaos theory, it might be argued that a perturbation acts on the scrambled set of a Li-Yorke chaotic map in a similar way, in the sense that they could exist specific sufficient conditions such that the scrambled set of a Li-Yorke map survives under a small perturbation. If this is true, it would be possible to deduce chaos for a discrete system in general through perturbative arguments.

\section{Numerical results}\label{sec_numerical}
In this section we focus on numerical results for the system \eqref{Ziegler_map_mod}. For all simulations, we return to the original parameters $\alpha$, $\tilde \alpha$, $\beta$, $\gamma$ as in \eqref{Ziegler_map_original}, remarking that $\sigma = -(\alpha + \tilde \alpha)$.

\subsection{Looking for a strange attractor}\label{sec_numerical_strange}
We present some numerical simulations for the orbit of the system \eqref{Ziegler_map_mod}.

Figure \ref{sym_xy} refers to the symmetric case $\tilde \alpha = -\alpha$, where we can notice the very sensitive dependence of the system on initial data.
\begin{figure}
\centering
\subfloat[][\label{sym_xy1}]
{\includegraphics[width=.3\textwidth]{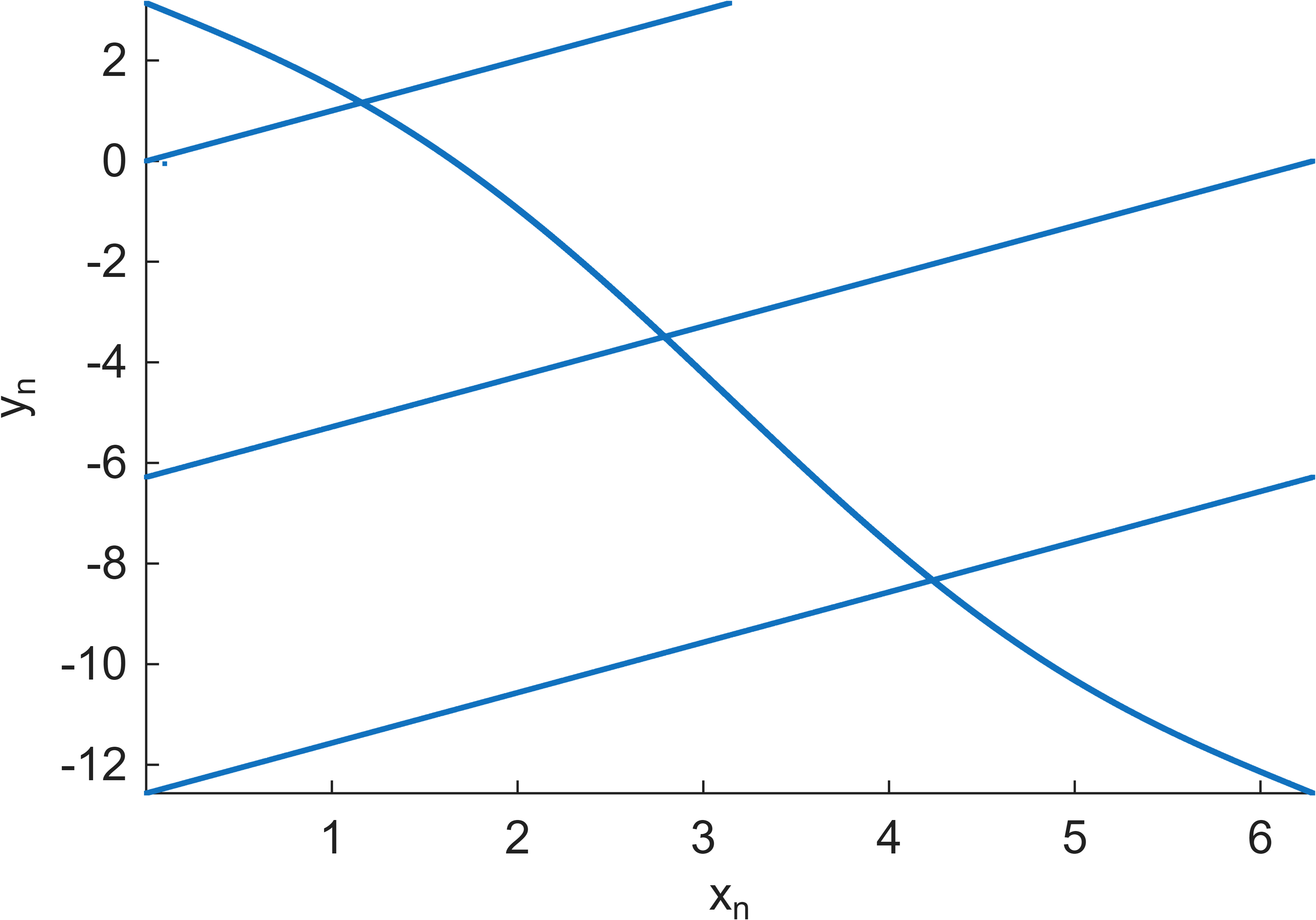}} \quad
\subfloat[][\label{sym_xy2}]
{\includegraphics[width=.3\textwidth]{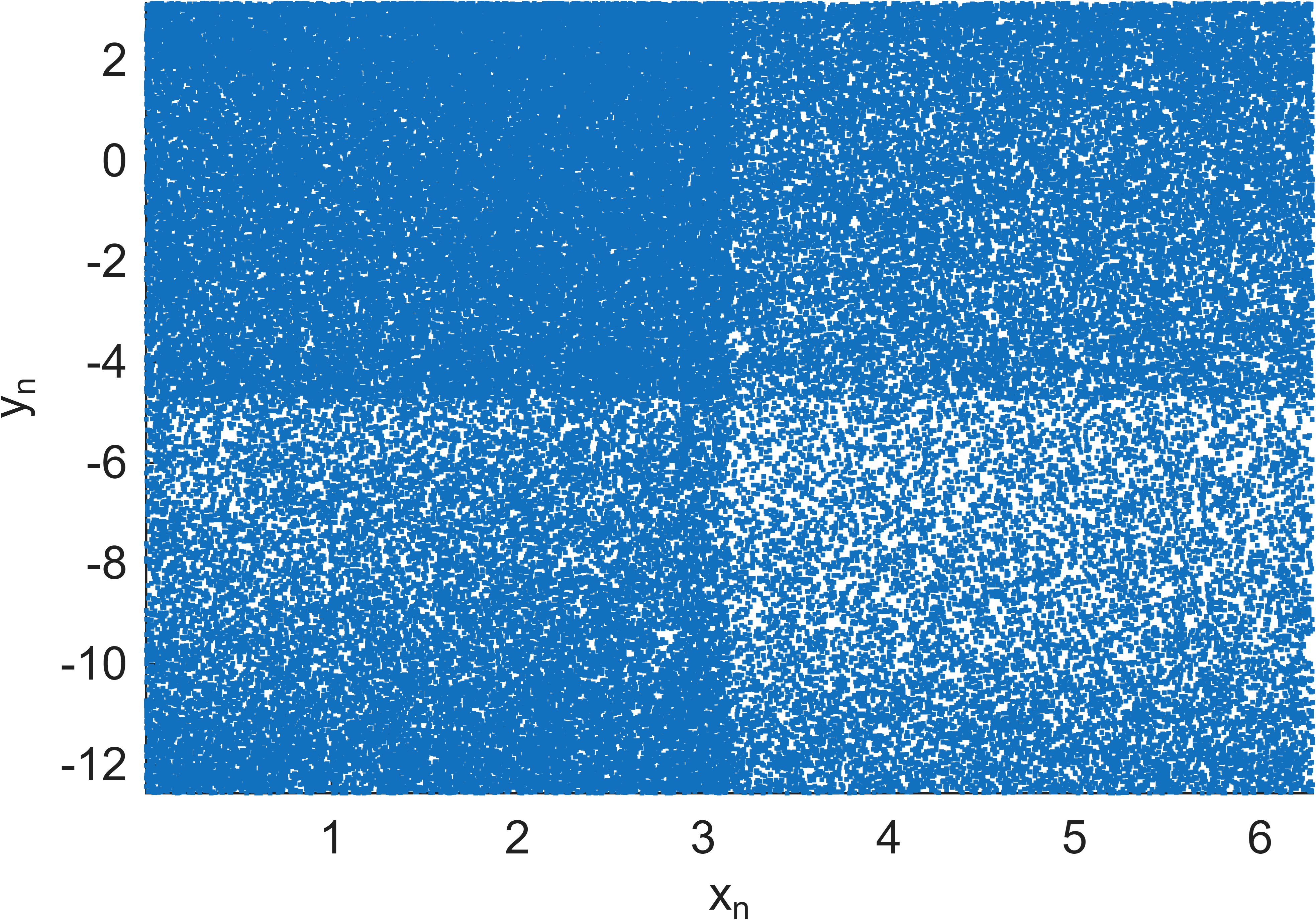}}
\caption{Projection onto the plane $(x, y)$, symmetric case $\tilde \alpha = -\alpha$. Parameters: $\alpha = -2$, $\tilde \alpha = 2$, $\beta = 0.5$, $\gamma = 1$. The simulation runs for $10^5$ iterations. a) Initial conditions: $x_0 = 0.1$, $y_0 = 0.1$, $z_0 = 0.1$, $\omega_0 = 0.1$. b) Initial conditions: $x_0 = 0.1001$, $y_0 = 0.1$, $z_0 = 0.1$, $\omega_0 = 0.1$.}
\label{sym_xy}
\end{figure}
In Figure \ref{sym_sameIC_xy} we set $\tilde \alpha = -\alpha$, all the initial conditions are set equal and we progressively increase the value of $\gamma$. This choice on parameters and initial conditions seems to be a particular regular case for the system. By increasing the value of $\gamma$, number of the straight trajectories increases, being present always in an odd quantity; furthermore, the motion is always restricted along these trajectories and a sinusoidal curve.
\begin{figure}
\centering
\subfloat[][\label{sym_sameIC_xy_gamma8}]
{\includegraphics[width=.3\textwidth]{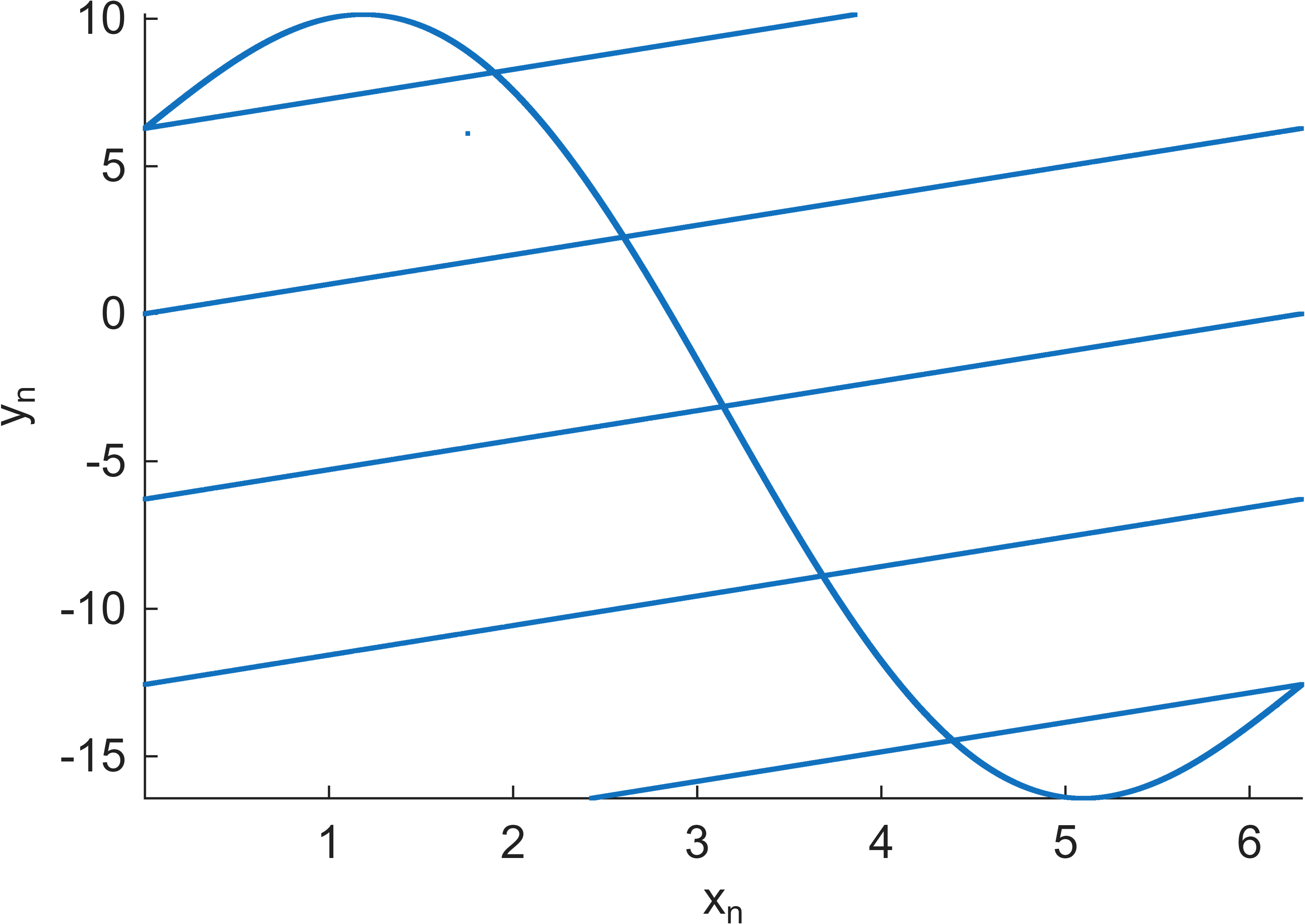}} \quad
\subfloat[][\label{sym_sameIC_xy_gamma13}]
{\includegraphics[width=.3\textwidth]{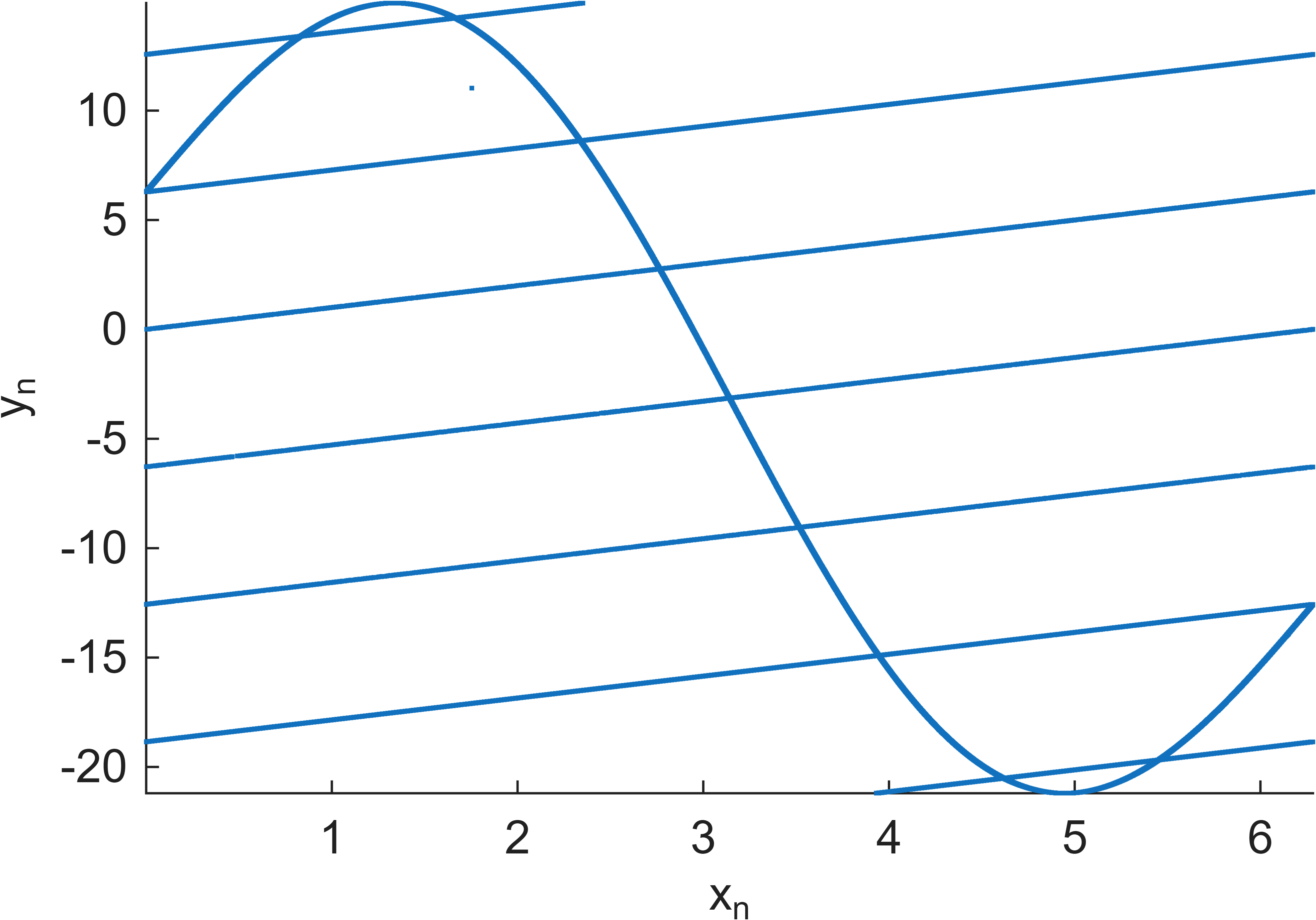}} \\
\subfloat[][\label{sym_sameIC_xy_gamma18}]
{\includegraphics[width=.3\textwidth]{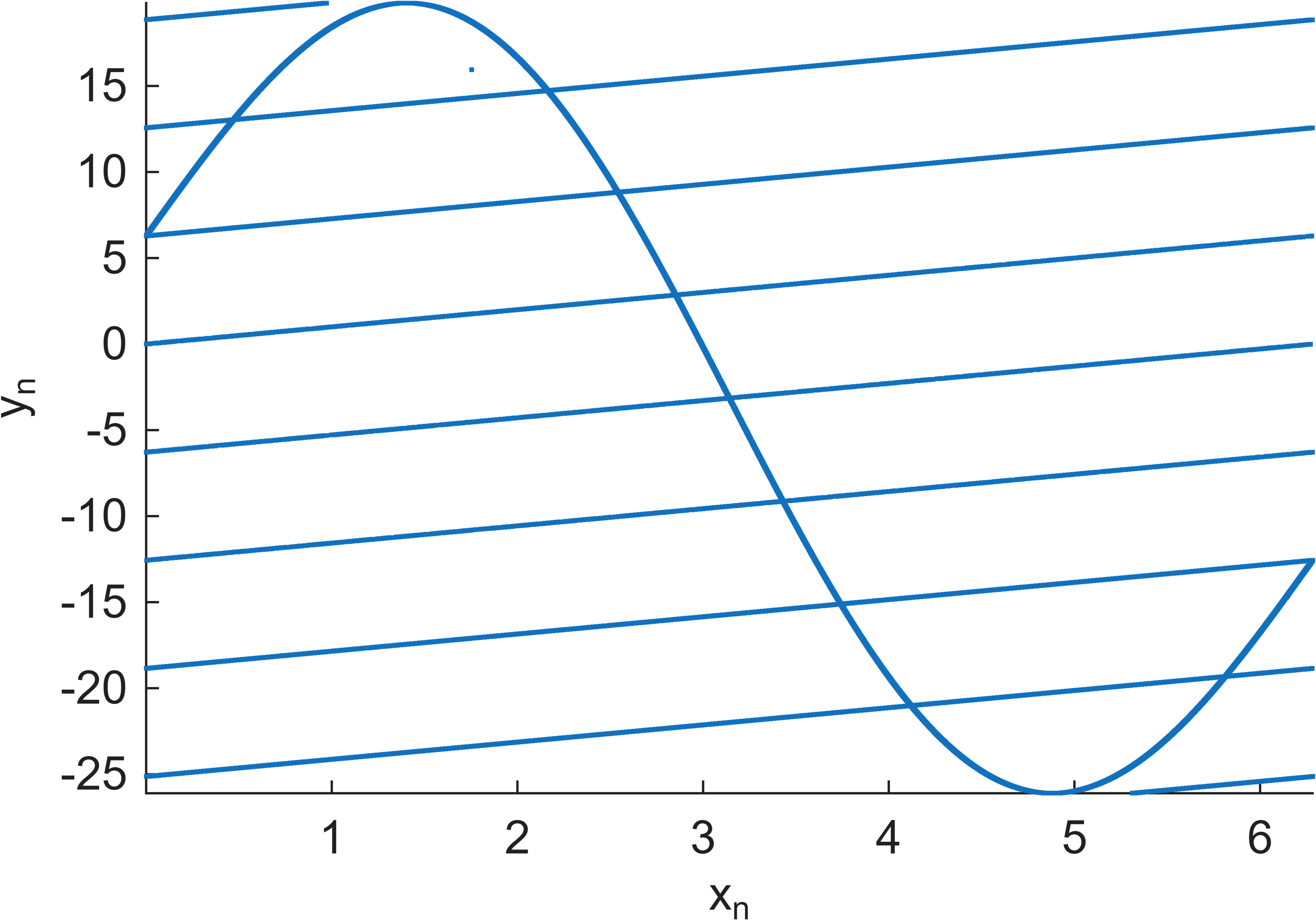}} \quad
\subfloat[][\label{sym_sameIC_xy_gamma23}]
{\includegraphics[width=.3\textwidth]{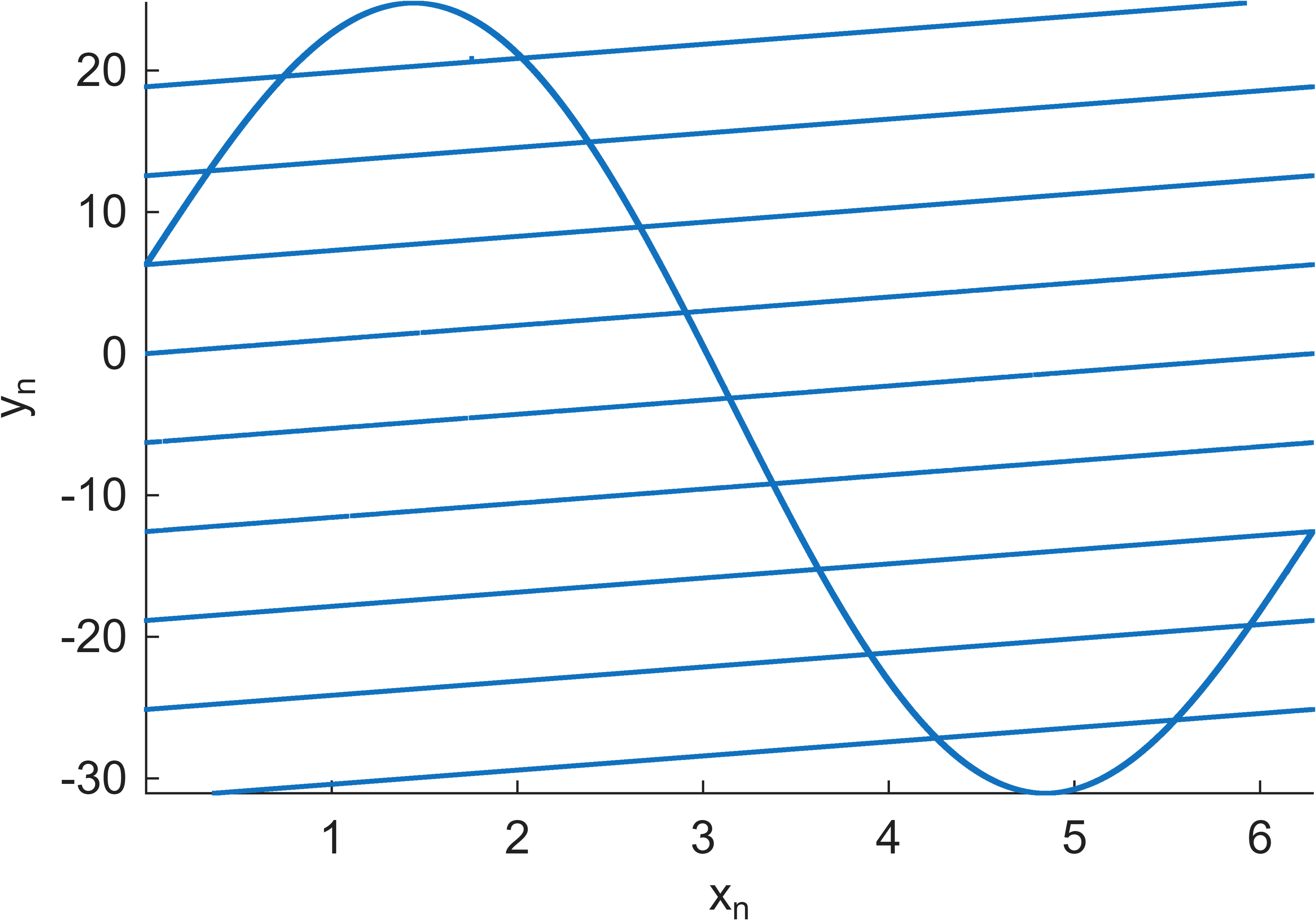}}
\caption{Projection onto the plane $(x, y)$, symmetric case $\tilde \alpha = -\alpha$. Parameters: $\alpha = -2$, $\tilde \alpha = 2$, $\beta = 1$. Initial conditions: $x_0 = 1.754$, $y_0 = 1.754$, $z_0 = 1.754$, $\omega_0 = 1.754$. The simulation runs for $10^5$ iterations. a) $\gamma = 8$. b) $\gamma = 13$. c) $\gamma = 18$. d) $\gamma = 23$.}
\label{sym_sameIC_xy}
\end{figure}

In Figure \ref{attractor_xz_tot} - \ref{attractor_yw_tot} we set $\tilde \alpha \ne \alpha$. It is recognized the arising of a strange attractor for the system, that exhibits a fractal structure analogous to the Hénon attractor \cite{Henon}. This attractor is confined to the projection planes $(x, z)$ and $(y, \omega)$, while it is replaced by completely random motion by taking projections on the planes $(x, y)$ and $(z, \omega)$, as shown in Figure \ref{attractor_xy_zw}.
\begin{figure}
\centering
\subfloat[][\label{attractor_xz}]
{\includegraphics[width=.85\textwidth]{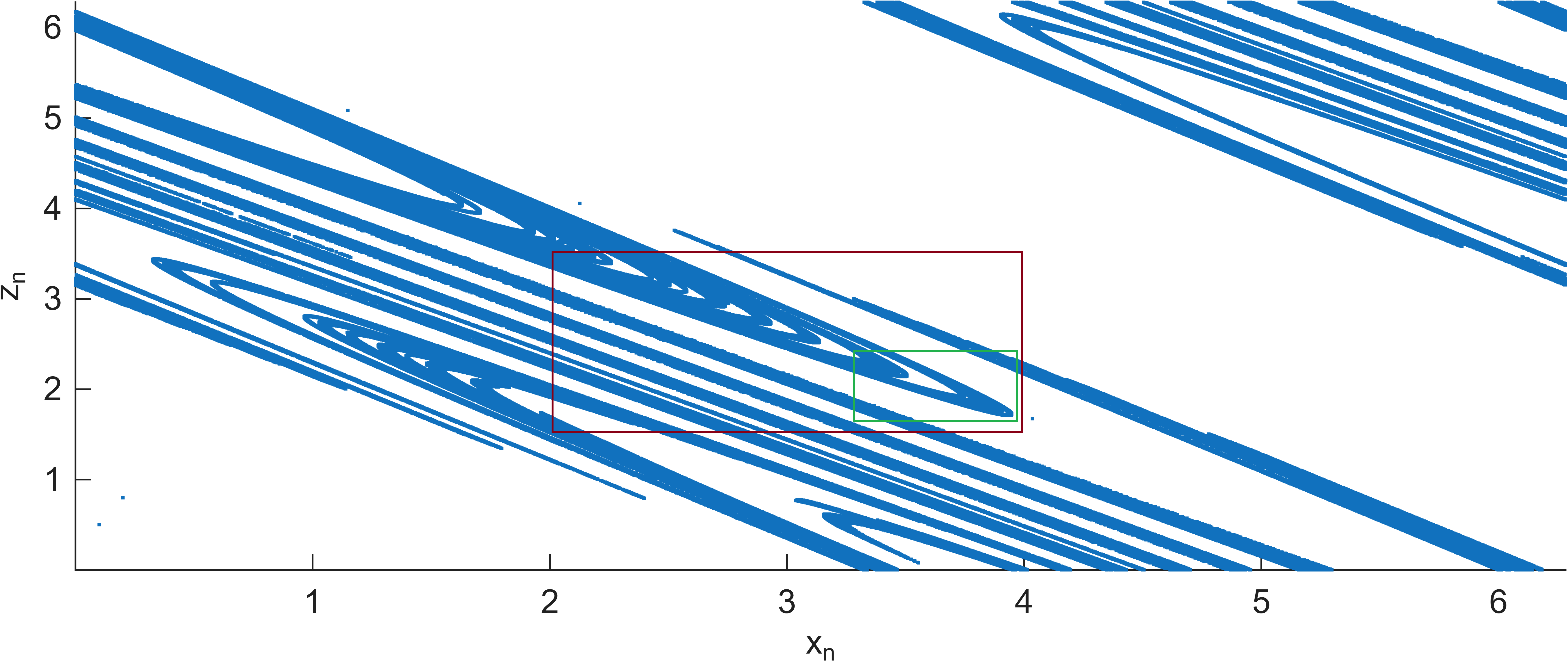}} \\
\subfloat[][\label{attractor_xz_zoom1}]
{\includegraphics[width=.85\textwidth]{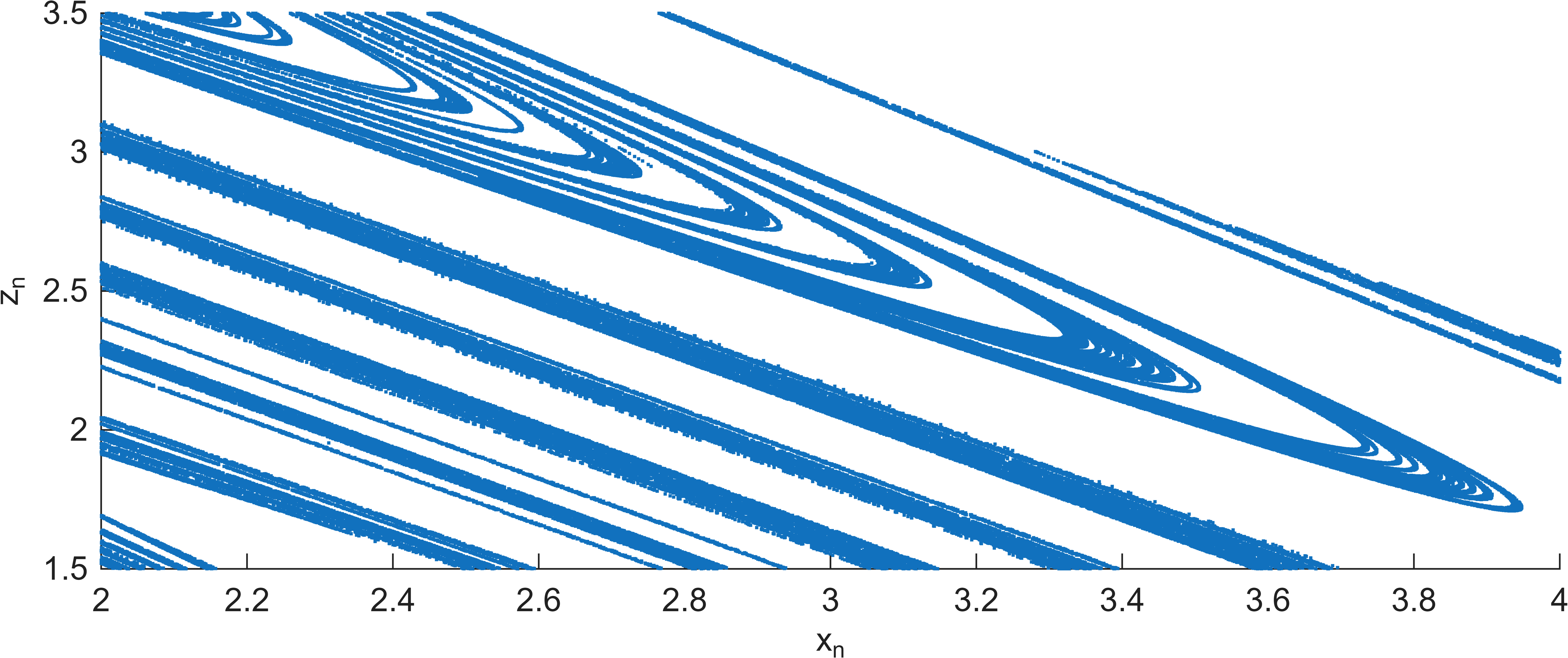}} \quad
\subfloat[][\label{attractor_xz_zoom2}]
{\includegraphics[width=.85\textwidth]{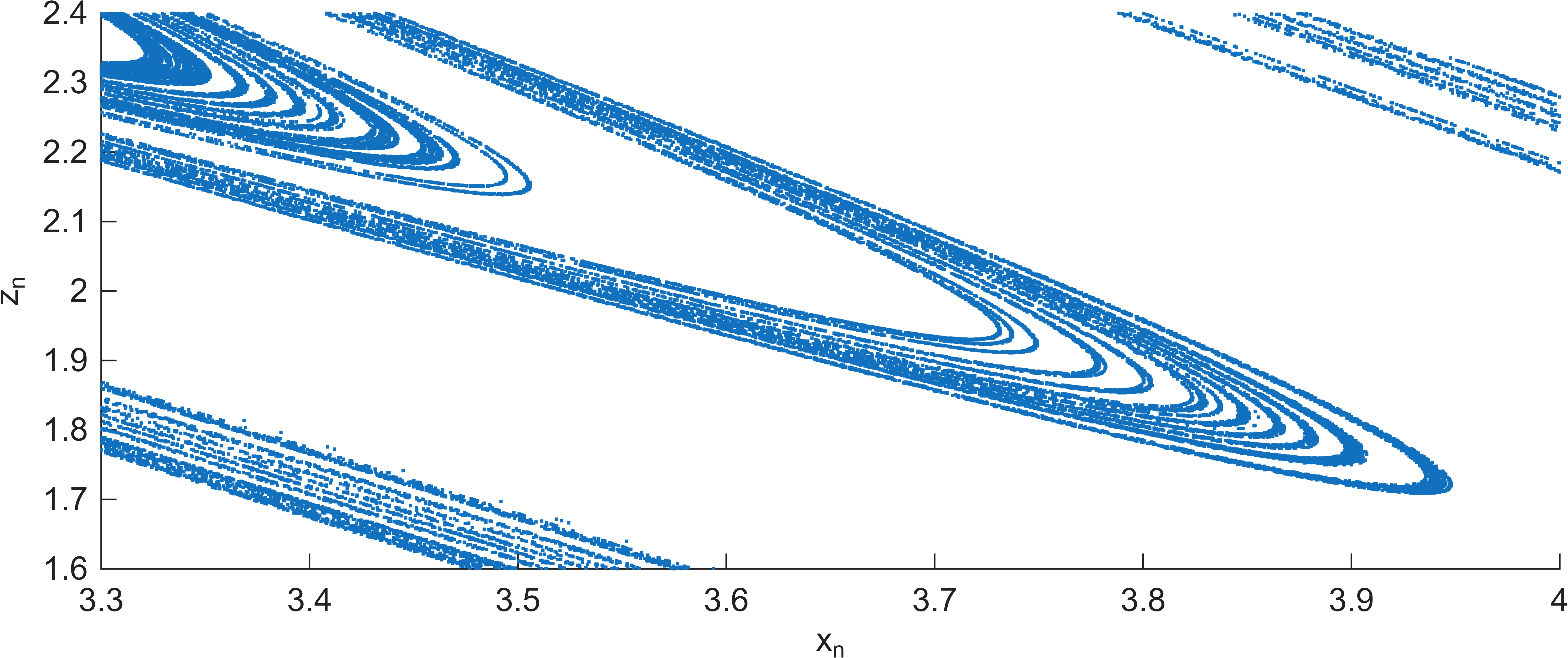}}
\caption{Confined strange attractor for the projection onto the plane $(x, z)$, non-symmetric case $\tilde \alpha \ne -\alpha$. The simulation runs for $10^6$ iterations. Parameters: $\alpha = -2.5$, $\tilde \alpha = 2$, $\beta = 0.8$, $\gamma = 10$. Initial conditions: $x_0 = 0.2$, $y_0 = 0.1$, $z_0 = 0.8$, $\omega_0 = 0.5$. a) Full attractor. b) Magnification of a) related to the red rectangle. c) Magnification of a) related to the green rectangle.}
\label{attractor_xz_tot}
\end{figure}
\begin{figure}
\centering
\subfloat[][\label{attractor_yw}]
{\includegraphics[width=.85\textwidth]{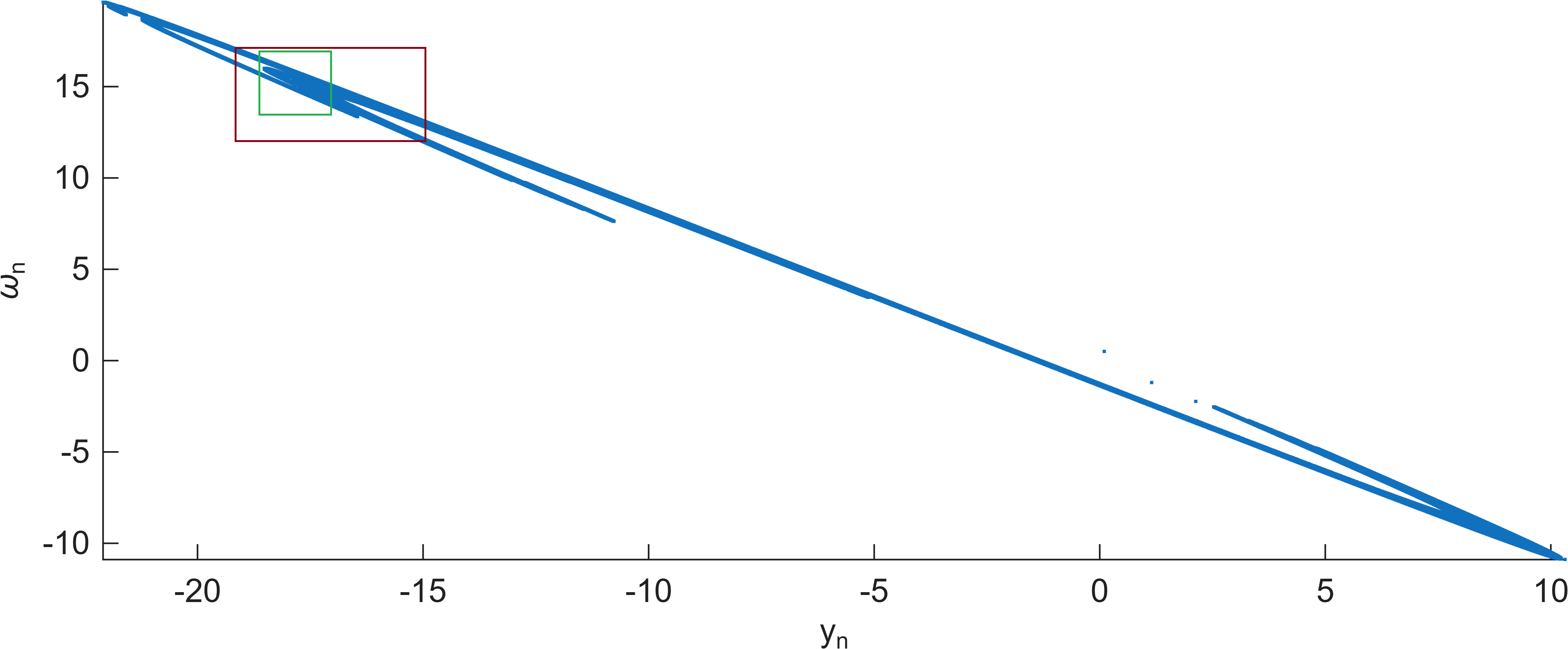}} \\
\subfloat[][\label{attractor_yw_zoom1}]
{\includegraphics[width=.85\textwidth]{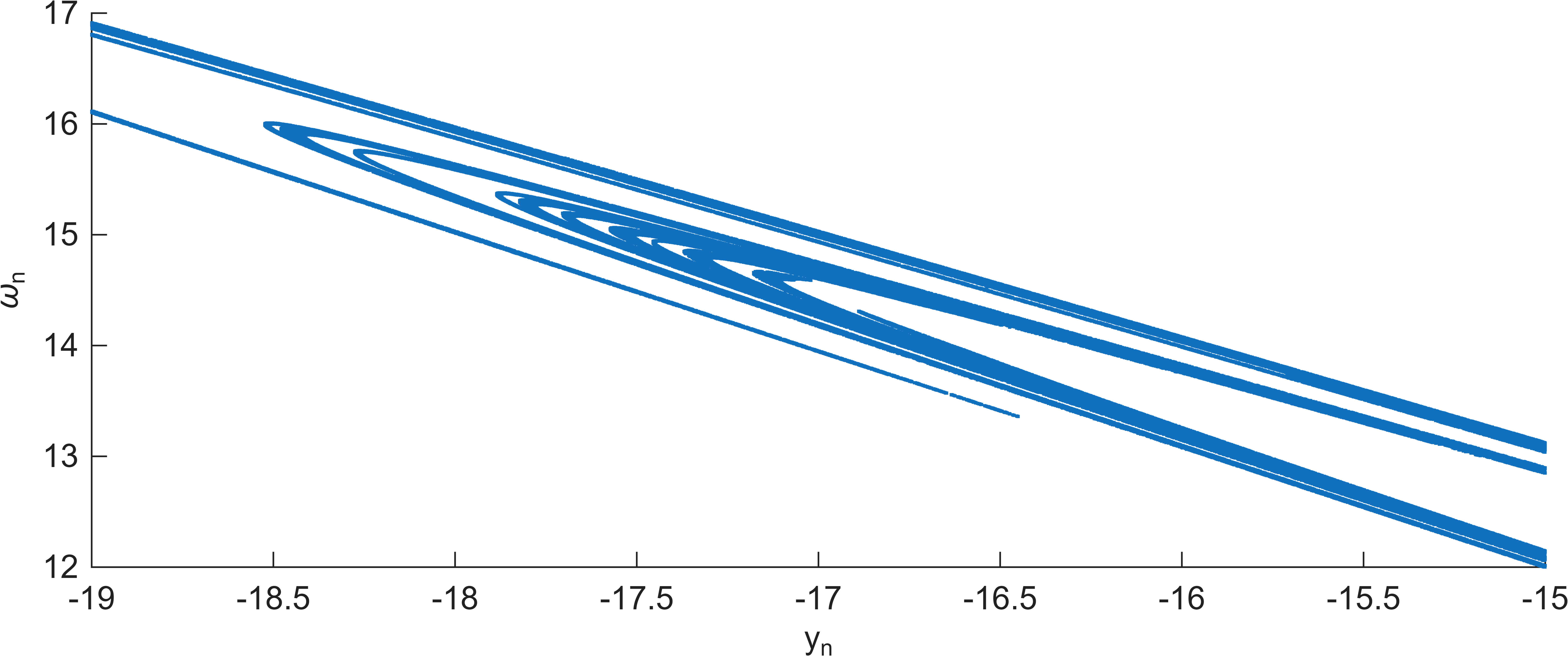}} \quad
\subfloat[][\label{attractor_yw_zoom2}]
{\includegraphics[width=.85\textwidth]{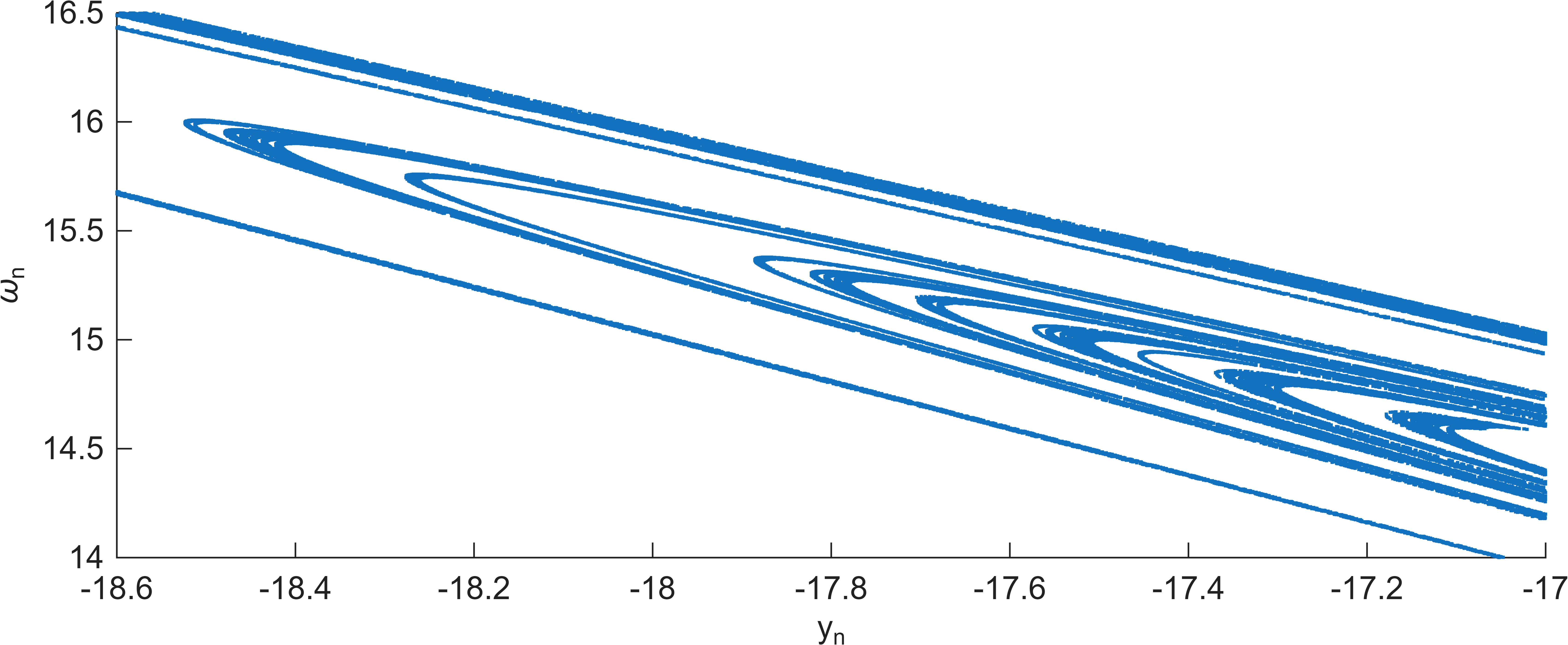}}
\caption{Confined strange attractor for the projection onto the plane $(y, \omega)$, non-symmetric case $\tilde \alpha \ne -\alpha$. The simulation runs for $10^6$ iterations. Parameters: $\alpha = -2.5$, $\tilde \alpha = 2$, $\beta = 0.8$, $\gamma = 10$. Initial conditions: $x_0 = 0.2$, $y_0 = 0.1$, $z_0 = 0.8$, $\omega_0 = 0.5$. a) Full attractor. b) Magnification of a) related to the red rectangle. c) Magnification of a) related to the green rectangle.}
\label{attractor_yw_tot}
\end{figure}
\begin{figure}
\centering
\subfloat[][\label{attractor_xy}]
{\includegraphics[width=.3\textwidth]{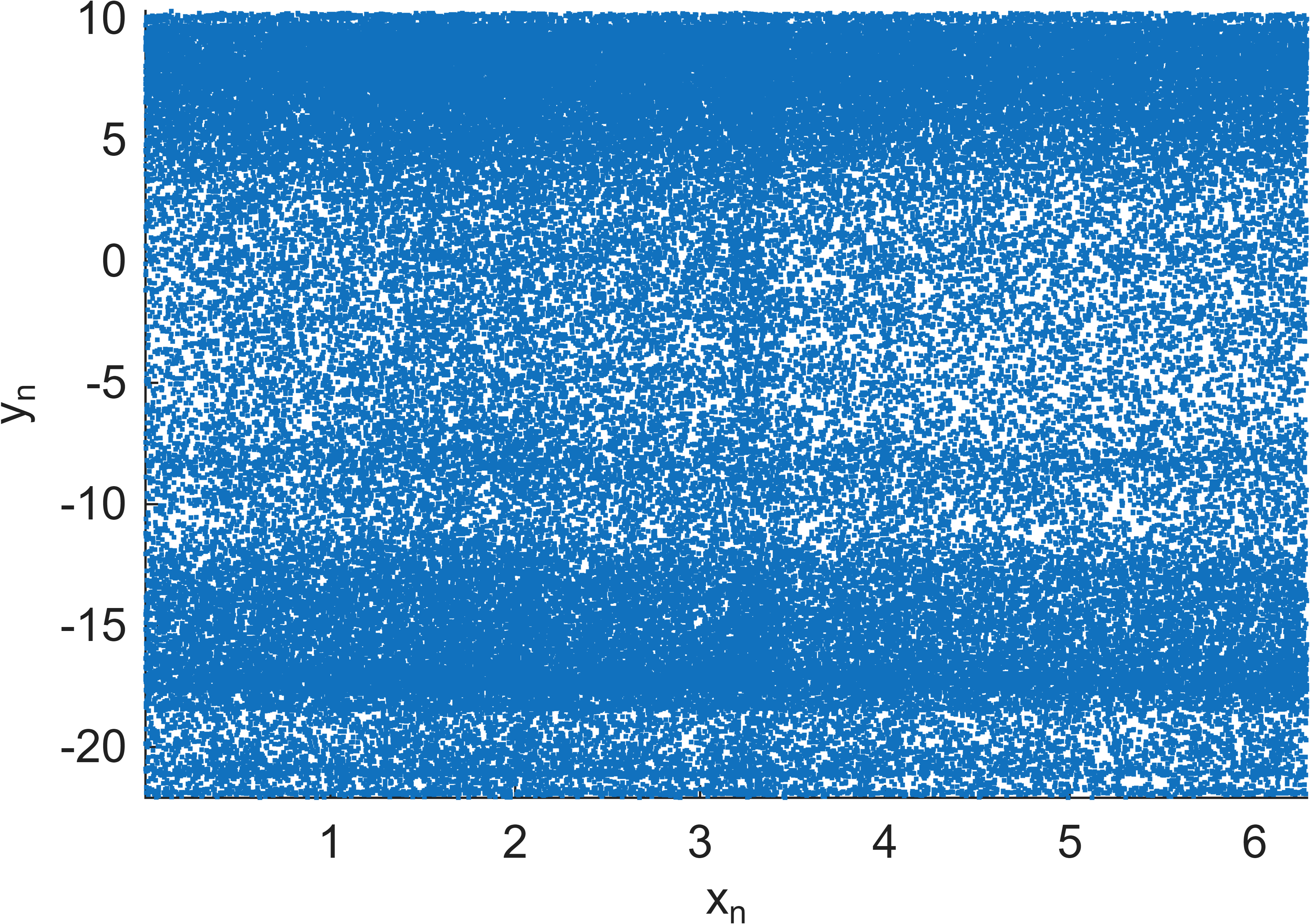}} \quad
\subfloat[][\label{attractor_zw}]
{\includegraphics[width=.3\textwidth]{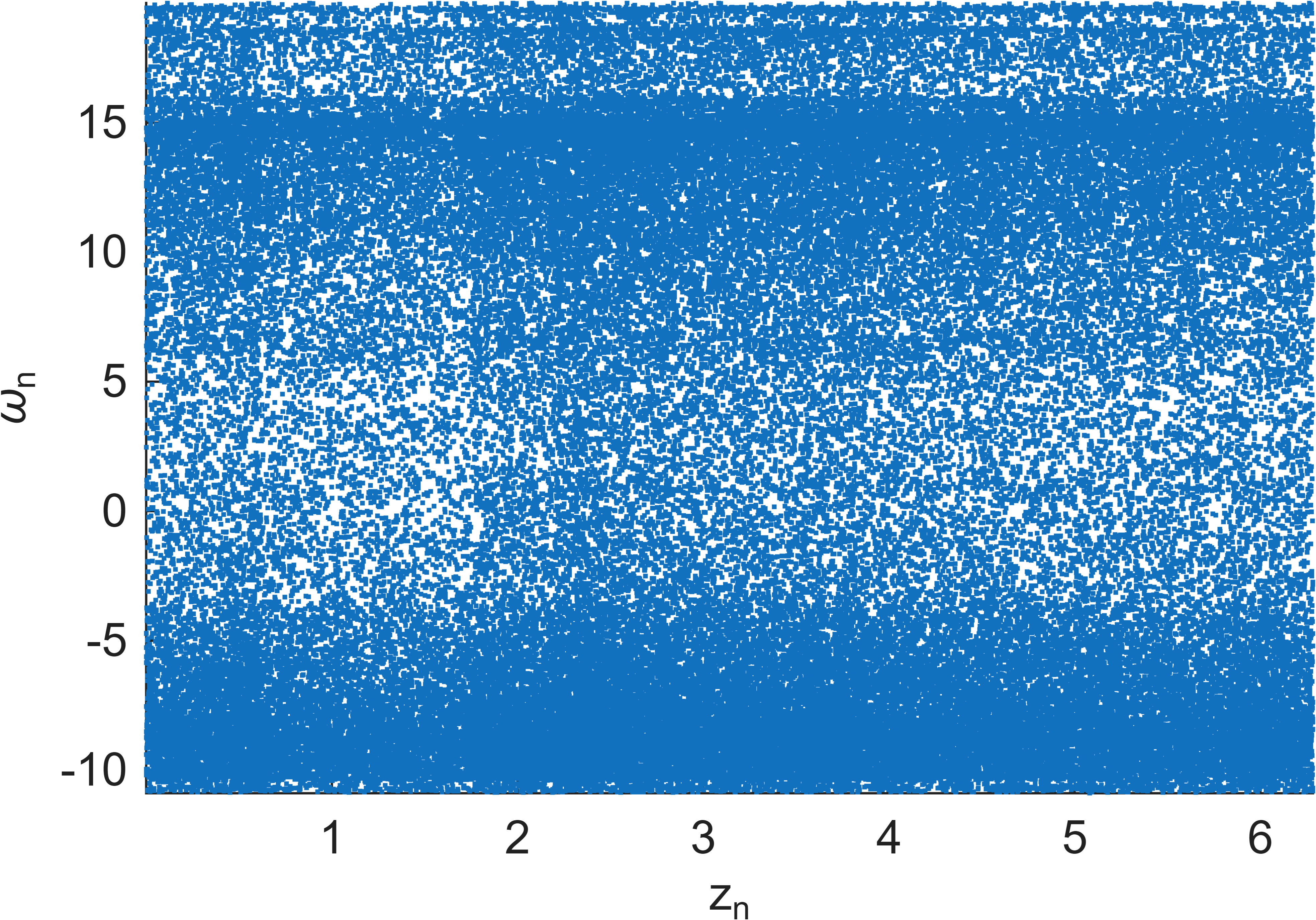}}
\caption{Random motion for projections on the planes $(x, y)$, $(z, \omega)$, non-symmetric case $\tilde \alpha \ne -\alpha$. The simulation shows $10^5$ iterations. Parameters: $\alpha = -2.5$, $\tilde \alpha = 2$, $\beta = 0.8$, $\gamma = 10$. Initial conditions: $x_0 = 0.2$, $y_0 = 0.1$, $z_0 = 0.8$, $\omega_0 = 0.5$. a) Projection on the plane $(x, y)$. b) Projection on the plane $(z, \omega)$.}
\label{attractor_xy_zw}
\end{figure}
\begin{rem}
Simulations shown in Figure \ref{attractor_xz_tot} - \ref{attractor_yw_tot} are not an isolated case. By varying parameters, the shape of the attractor changes; however, the attractor is always confined to the coordinate-coordinate and momentum-momentum projection planes, while it is replaced by completely random motion on the coordinate-momentum projection plane. For fixed values of parameters, the attractor survives for almost every sets of initial conditions.
\end{rem}
Now, we briefly focus on the case $\gamma = 0$. By recalling the physical origin of the system \eqref{Ziegler_map}, the term $\gamma$ represents an external non-conservative force acting on an Hamiltonian system; by setting $\gamma = 0$, the map comes out from an integrable continuous system whose phase space is foliated in invariant tori, due to Liouville-Arnold theorem \cite{Liouville, Arnold1}. Then, it is worth to see whether the strange attractor shown in Figure \ref{attractor_xz_tot} - \ref{attractor_yw_tot} survives for this unperturbed case.

In Figure \ref{attractor_xz_gamma0_tot} - \ref{attractor_yw_gamma0_tot} it is shown the orbit of the system \eqref{Ziegler_map_mod}; remaining parameters and initial conditions are taken as in Figure \ref{attractor_xz_tot} - \ref{attractor_yw_tot}. We can see that the strange attractor survives maintaining a fractal structure; however, the absence of the sinusoidal term due to a vanishing value for $\gamma$ removes the curved trajectories present in Figure \ref{attractor_xz_tot} - \ref{attractor_yw_tot}.
\begin{figure}
\centering
\subfloat[][\label{attractor_xz_gamma0}]
{\includegraphics[width=.3\textwidth]{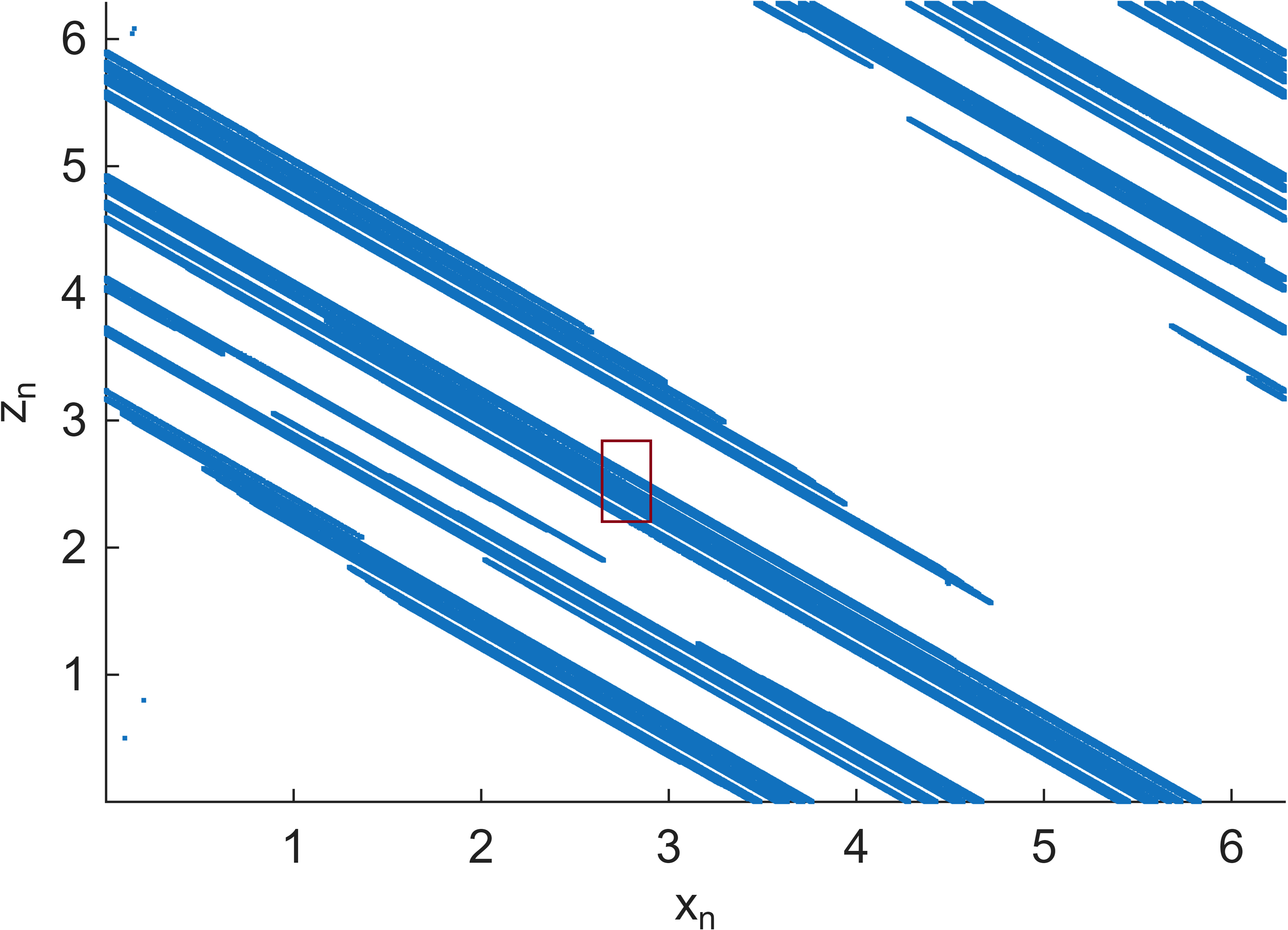}} \quad
\subfloat[][\label{attractor_xz_gamma0_zoom}]
{\includegraphics[width=.3\textwidth]{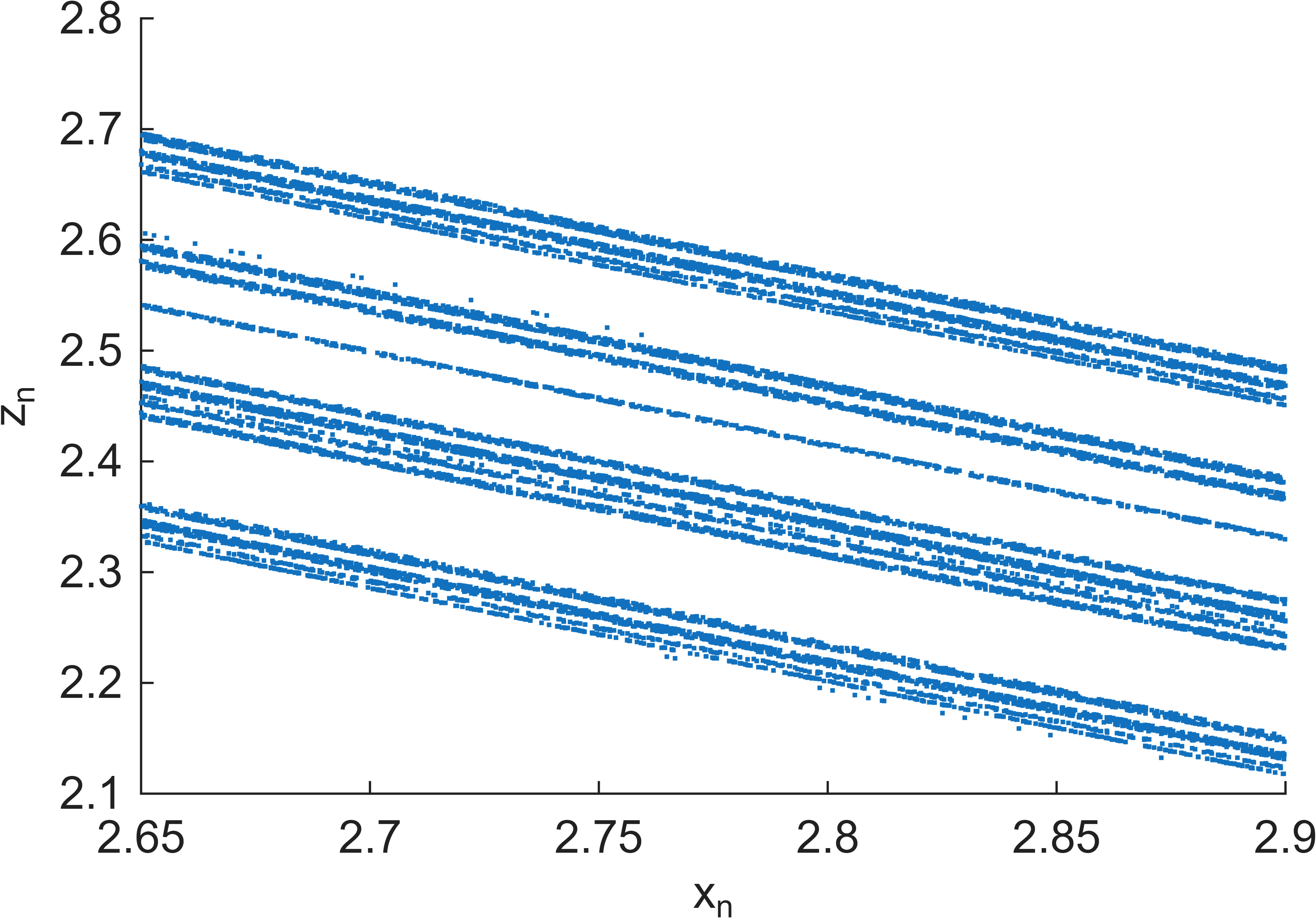}}
\caption{Confined strange attractor for the projection onto the plane $(x, z)$, non-symmetric case $\tilde \alpha \ne -\alpha$. The simulation runs for $10^6$ iterations. Parameters: $\alpha = -2.5$, $\tilde \alpha = 2$, $\beta = 0.8$, $\gamma = 0$. Initial conditions: $x_0 = 0.2$, $y_0 = 0.1$, $z_0 = 0.8$, $\omega_0 = 0.5$. a) Full attractor. b) Magnification of a) related to the red rectangle.}
\label{attractor_xz_gamma0_tot}
\end{figure}
\begin{figure}
\centering
\subfloat[][\label{attractor_yw_gamma0}]
{\includegraphics[width=.3\textwidth]{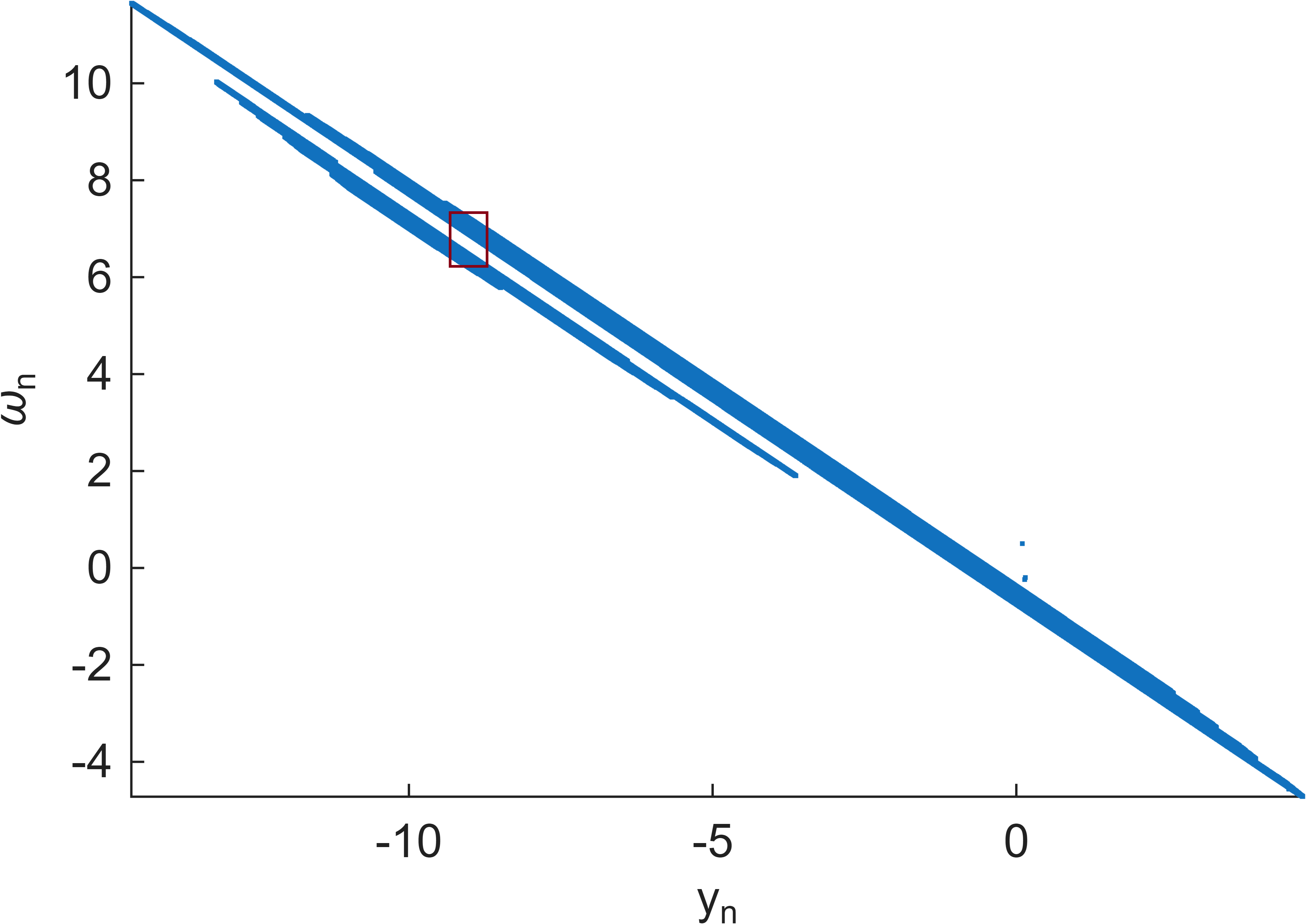}} \quad
\subfloat[][\label{attractor_yw_gamma0_zoom}]
{\includegraphics[width=.3\textwidth]{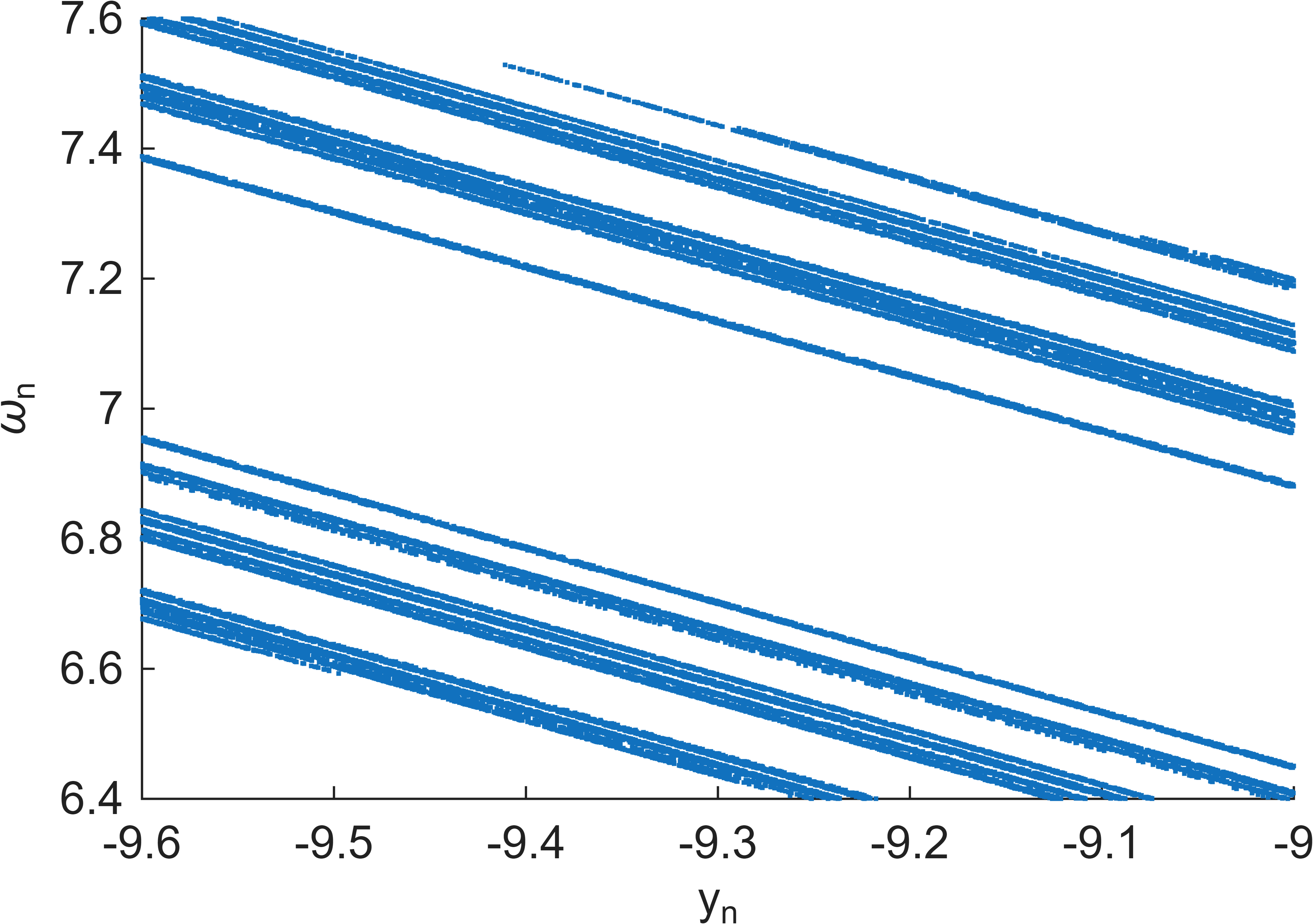}}
\caption{Confined strange attractor for the projection onto the plane $(y, \omega)$, non-symmetric case $\tilde \alpha \ne -\alpha$. The simulation runs for $10^6$ iterations. Parameters: $\alpha = -2.5$, $\tilde \alpha = 2$, $\beta = 0.8$, $\gamma = 0$. Initial conditions: $x_0 = 0.2$, $y_0 = 0.1$, $z_0 = 0.8$, $\omega_0 = 0.5$. a) Full attractor. b) Magnification of a) related to the red rectangle.}
\label{attractor_yw_gamma0_tot}
\end{figure}

In Figure \ref{attractor_2orbits} they are shown projections on the planes $(x, z)$ and $(y, \omega)$ of two trajectories with very close initial conditions, giving a further confirmation of the sensitive dependence on initial data for the system. The sensitive dependence of the system shown in Figure \ref{attractor_2orbits} will found another confirm in Section \ref{sec_numerical_bifurcation} through the computation of positive Lyapunov exponents.
\begin{figure}
\centering
\subfloat[][\label{attractor_xz_2orbits}]
{\includegraphics[width=.4\textwidth]{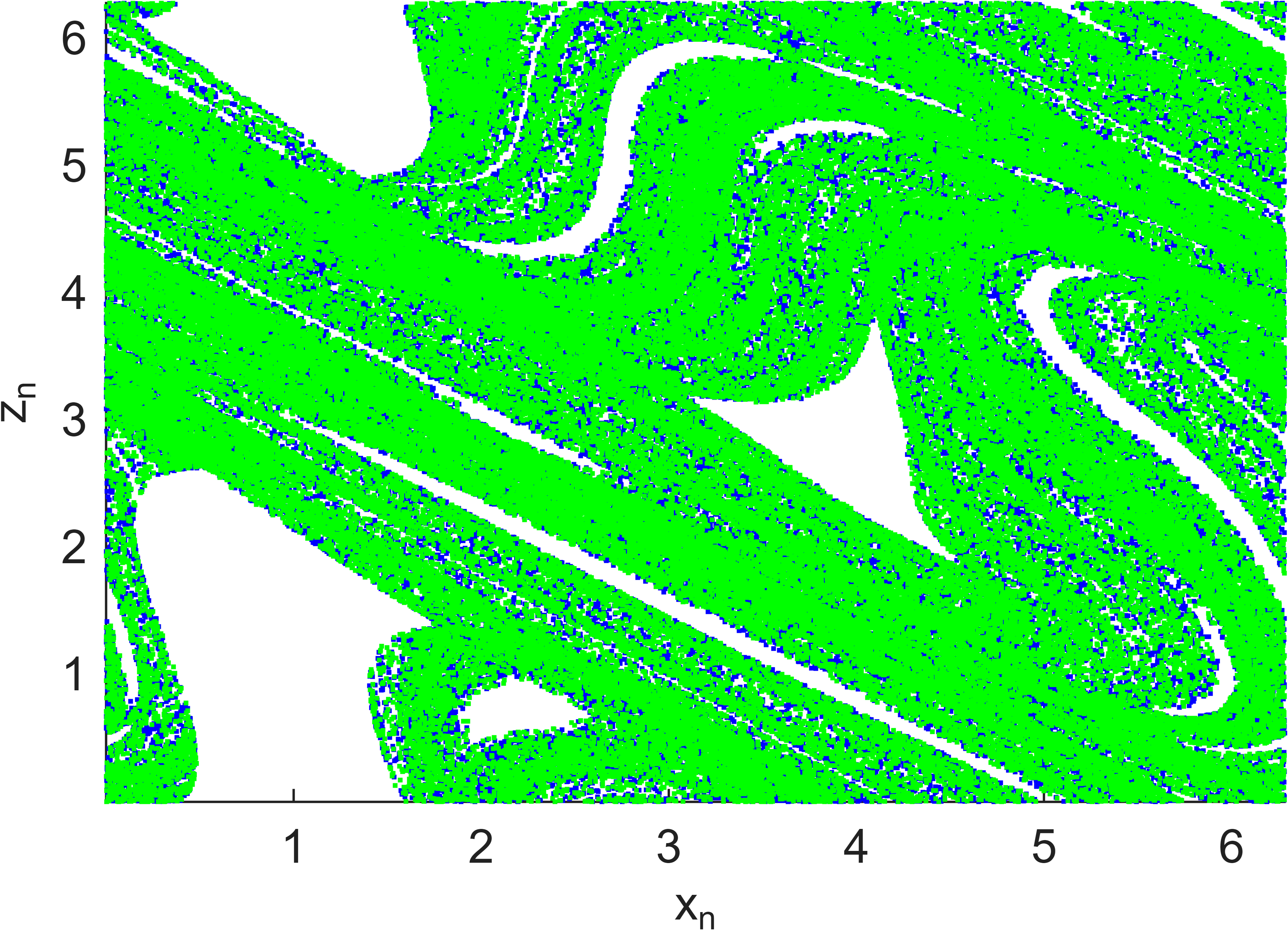}} \quad
\subfloat[][\label{attractor_yw_2orbits}]
{\includegraphics[width=.4\textwidth]{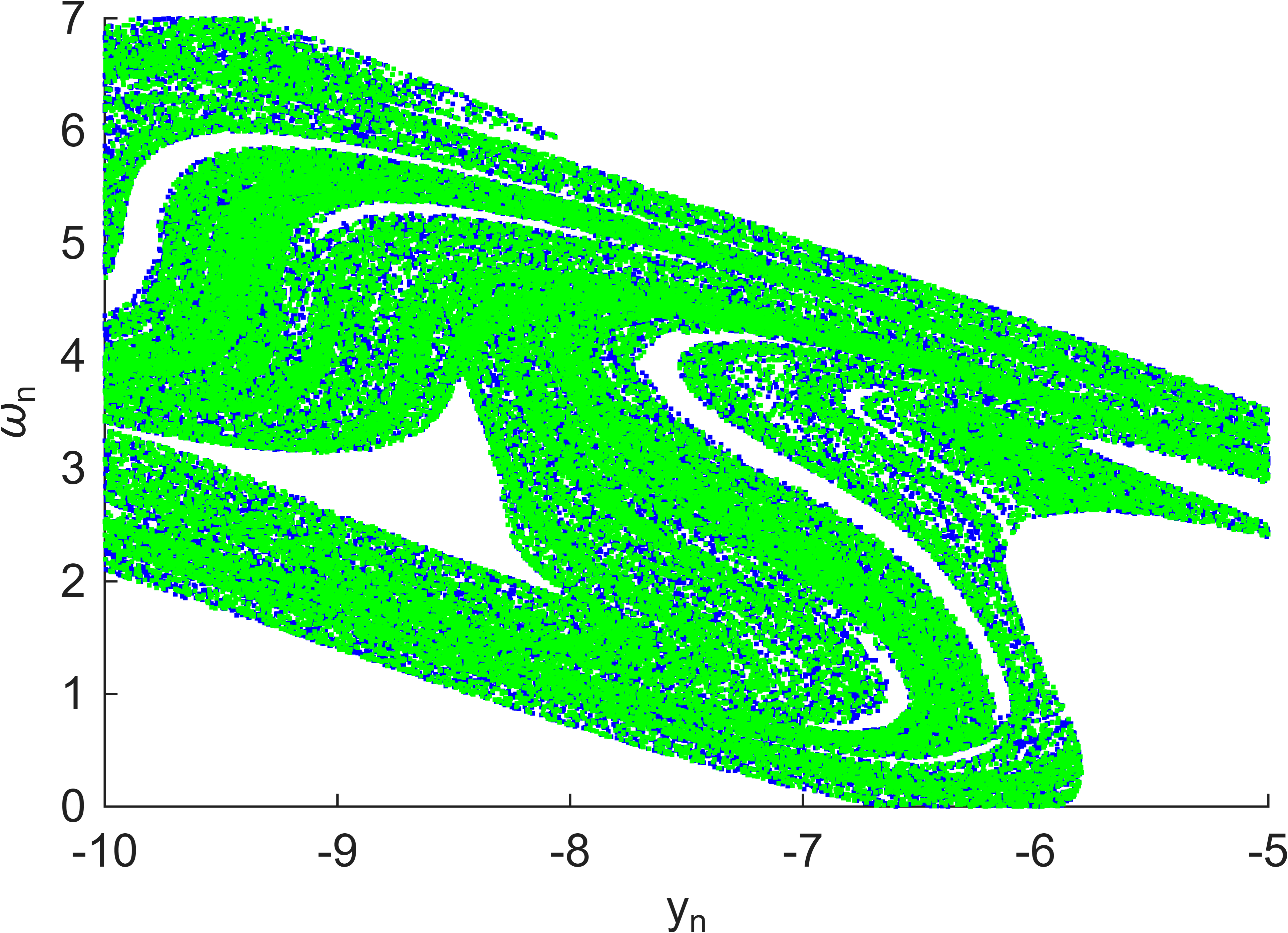}}
\caption{Confined strange attractor for two very close initial conditions. Parameters: $\alpha = -3.5$, $\tilde \alpha = 2$, $\beta = 0.8$, $\gamma = -4$. Initial conditions, blue orbit: $x_0 = 1.4$, $y_0 = 0.8$, $z_0 = 2.7$, $\omega_0 = 2.34$. Initial conditions, green orbit: $x_0 = 1.4001$, $y_0 = 0.7999$, $z_0 = 2.7001$, $\omega_0 = 2.3399$. The simulation runs for $10^5$ iterations. a) Projection onto the plane $(x, z)$. b) Projection onto the plane $(y, \omega)$, magnification on the center.}
\label{attractor_2orbits}
\end{figure}

\subsection{Bifurcation diagrams and Lyapunov exponents}\label{sec_numerical_bifurcation}
In this section we perform a bifurcation analysis of the system \eqref{Ziegler_map_mod}. We start with the following simple result.
\begin{prop}\label{prop_1D_map}
For $\tilde \alpha = -\alpha$, the dynamics of \eqref{Ziegler_map_mod} is uniquely determined by $y_0$, $y_1$.
\end{prop}
\begin{proof}
We set $\tilde \alpha = -\alpha$. By substituting $z_n$, $\omega_n$ in the first two equations of \eqref{Ziegler_map_mod}, the system reduces to
\beq\begin{cases}
y_{n+1} = f( \Mod(y_{n-1} \,, 2\pi) \,, \Mod(-y_{n-1} \,, 2\pi) ) \\
y_{n+2} = f( \Mod(y_n \,, 2\pi) \,, \Mod(-y_n \,, 2\pi) ) \,.
\end{cases}\eeq
Therefore, \eqref{Ziegler_map_mod} is equivalent to
\beq\label{1D_map}
y_{n+2} = f(y_n) \,, \quad f(y_n) = \alpha \Mod(y_n \,, 2\pi) + \beta \Mod(-y_n \,, 2\pi) + \gamma \sin( \Mod(y_n \,, 2\pi) ) \,.
\eeq
Given $(y_0, y_1)$, the orbit is uniquely determined, i.e. $(y_0 \,, y_1) \to (y_2 \,, y_3 \,, y_4 \,, \dots)$. In particular, we can split the overall orbit into an even and an odd orbit, that depend uniquely on $y_0$ and $y_1$, respectively, i.e. for $k \in \bb{N}$:
\beq\begin{split}
&y_0 \to (y_2 \,, y_4 \,, \dots, y_{2k} \,, \dots) \,, \\
&y_1 \to (y_3 \,, y_5 \,, \dots, y_{2k+1} \,, \dots) \,.
\end{split}\eeq
\end{proof}
Given Proposition \ref{prop_1D_map}, we begin performing the bifurcation diagram of \eqref{1D_map}. We implement the classical algorithm \cite{Strogatz} for the construction of the bifurcation diagram. The basic idea of the algorithm is to compute the trajectory of the system for every value of the bifurcation parameter up to a certain iteration step ($n_{it}$), discarding the first iterations ($n_{tran}$) so that the system has actually reached its attractor and any transitory effects are not shown. In particular, we construct two separate lists $y_{even}$, $y_{odd}$ for the even and odd iterations respectively; then, we alternatively plot the values of $y_{even}$ and $y_{odd}$ in order to obtain the whole orbit.

In Figure \ref{1Dbif_gamma} it is shown the bifurcation diagram $(\gamma, y_n)$ associated to the map \eqref{1D_map} for $\tilde \alpha = -\alpha = -1$, $\beta = 0.1$, $\gamma \in [0, 5]$; we perform $n_{it} = 3.000$ iterations by discarding the first $n_{tran} = 300$ ones and discretizing $\gamma$ into $10.000$ values.
\begin{figure}
\centering
\includegraphics[width=10cm]{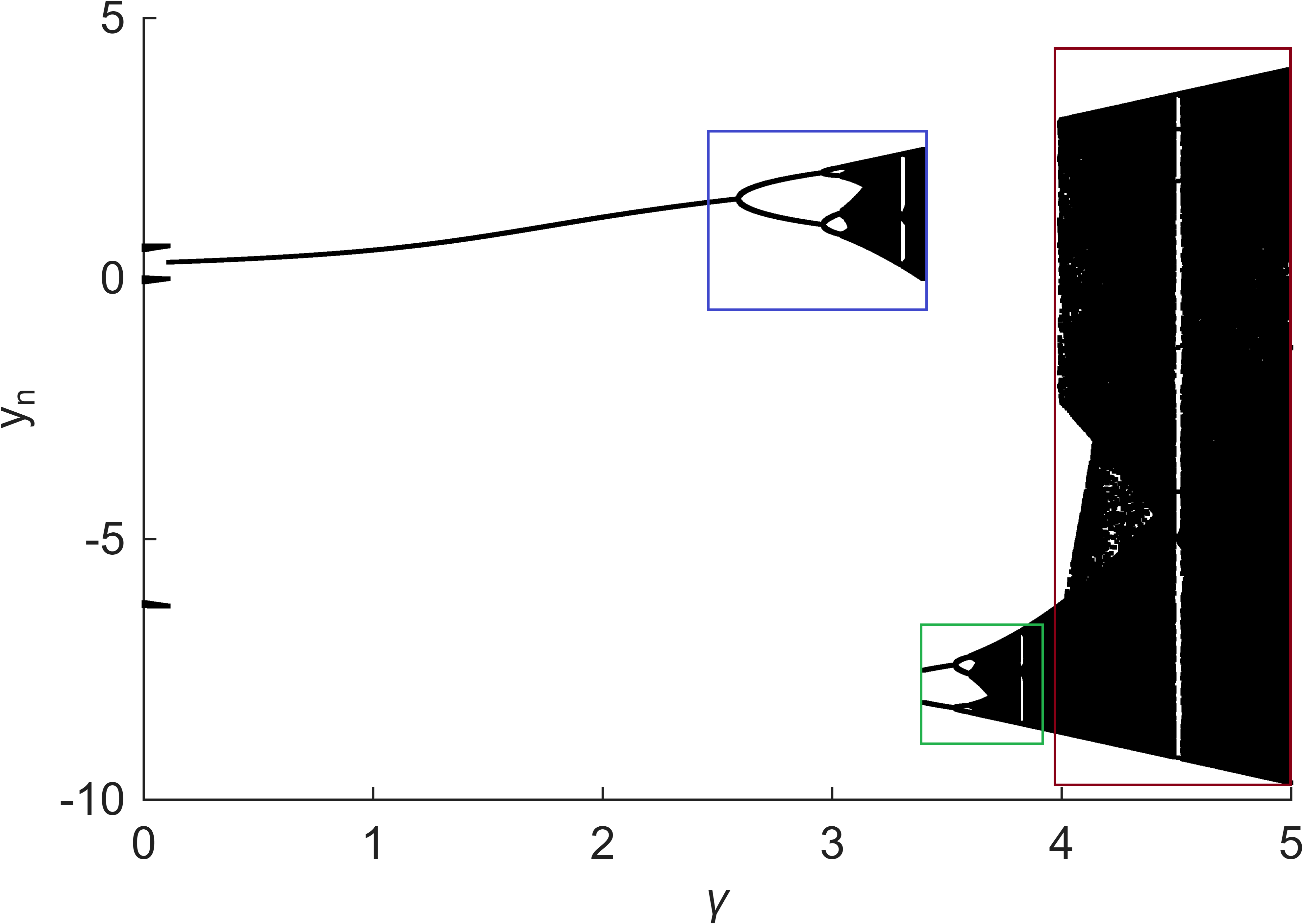}
\caption{Bifurcation diagram $(\gamma, y_n)$ for the map \eqref{1D_map}. Parameters: $\tilde \alpha = -\alpha = 1$, $\beta = 0.1$.}
\label{1Dbif_gamma}
\end{figure}
\begin{figure}
\centering
\includegraphics[width=7cm]{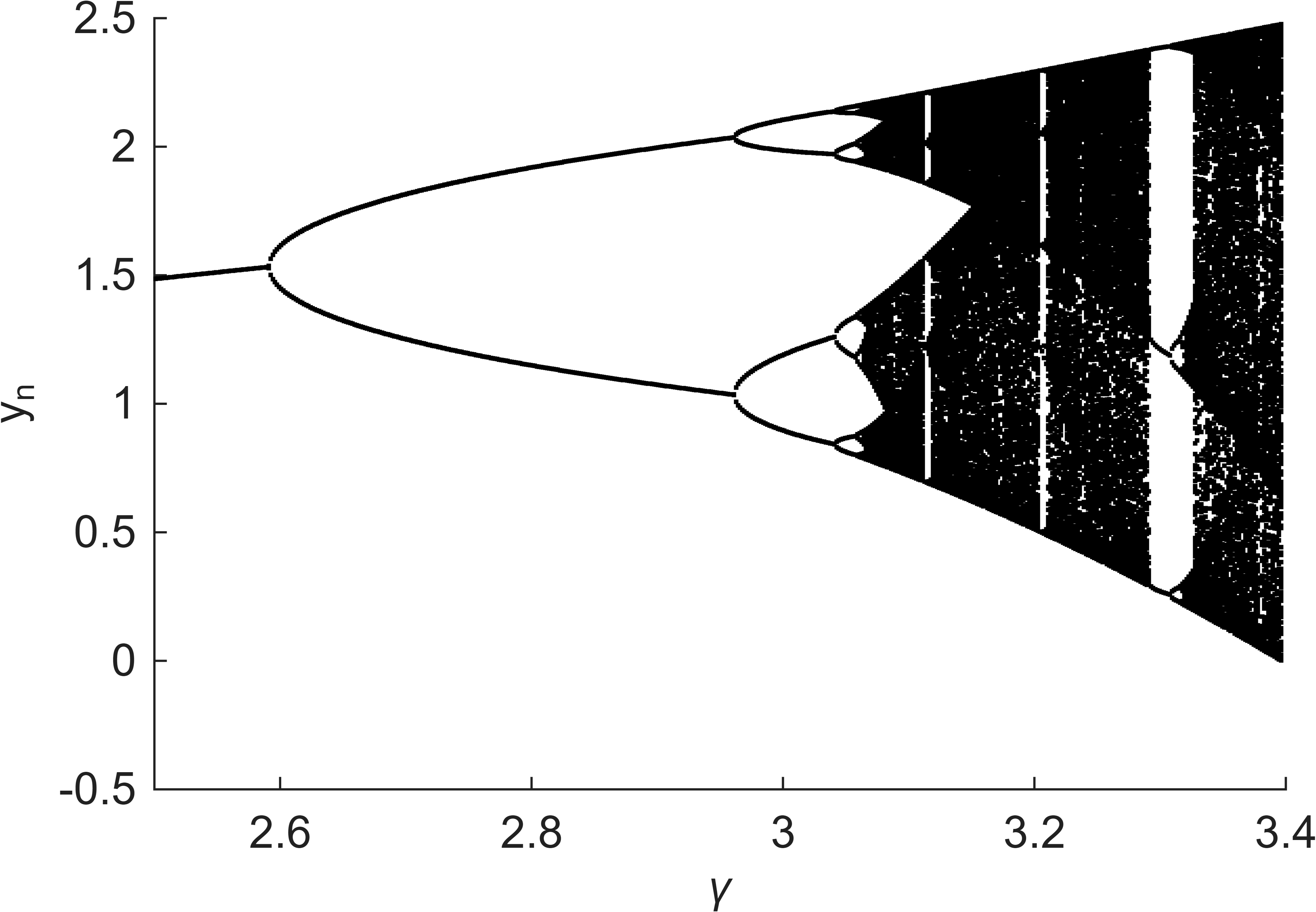}
\caption{Magnification of Figure \ref{1Dbif_gamma} related to the blue rectangle.}
\label{1Dbif_gamma_zoom1}
\end{figure}
\begin{figure}
\centering
\includegraphics[width=7cm]{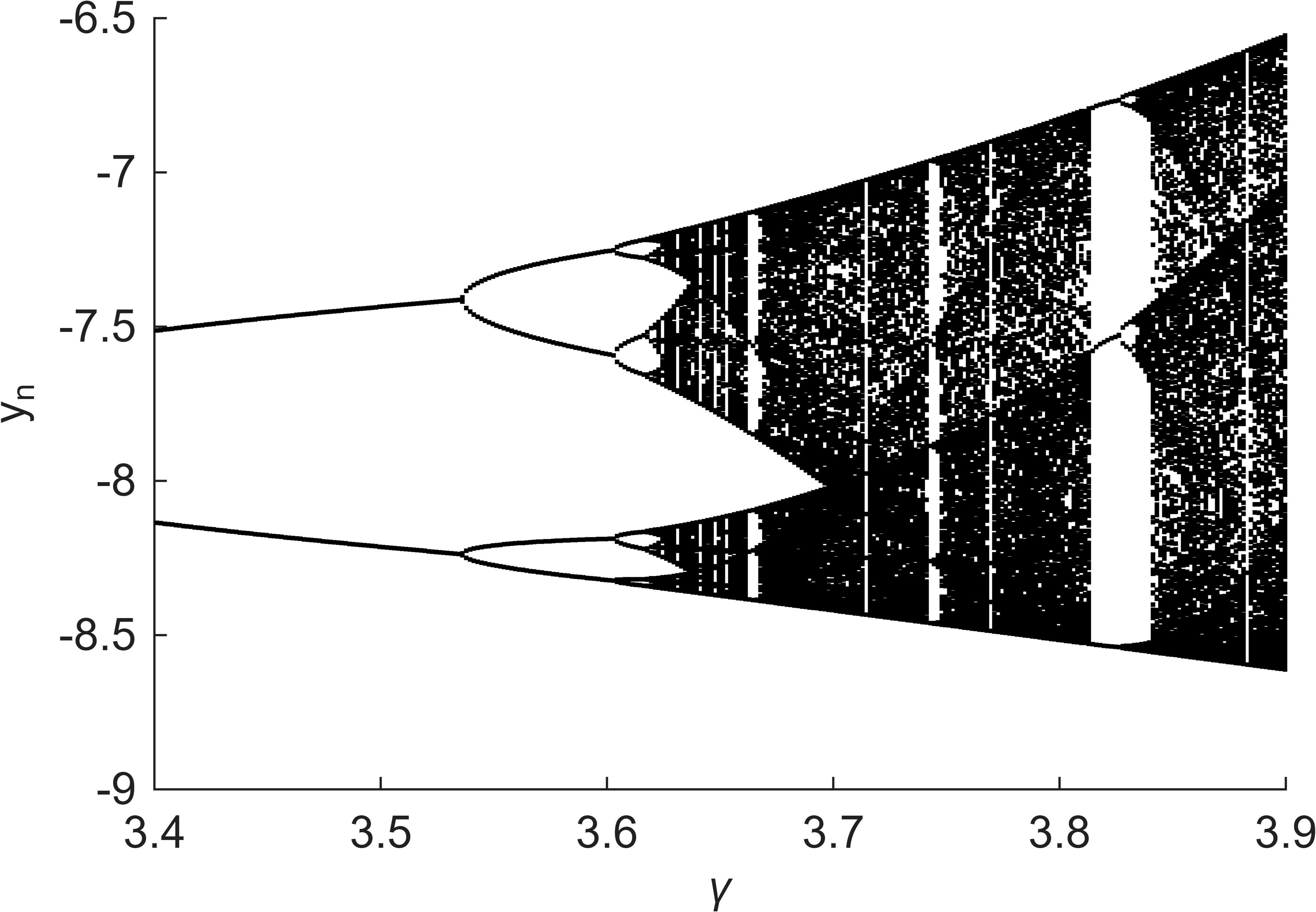}
\caption{Magnification of Figure \ref{1Dbif_gamma} related to the green rectangle.}
\label{1Dbif_gamma_zoom2}
\end{figure}
\begin{figure}
\centering
\includegraphics[width=7cm]{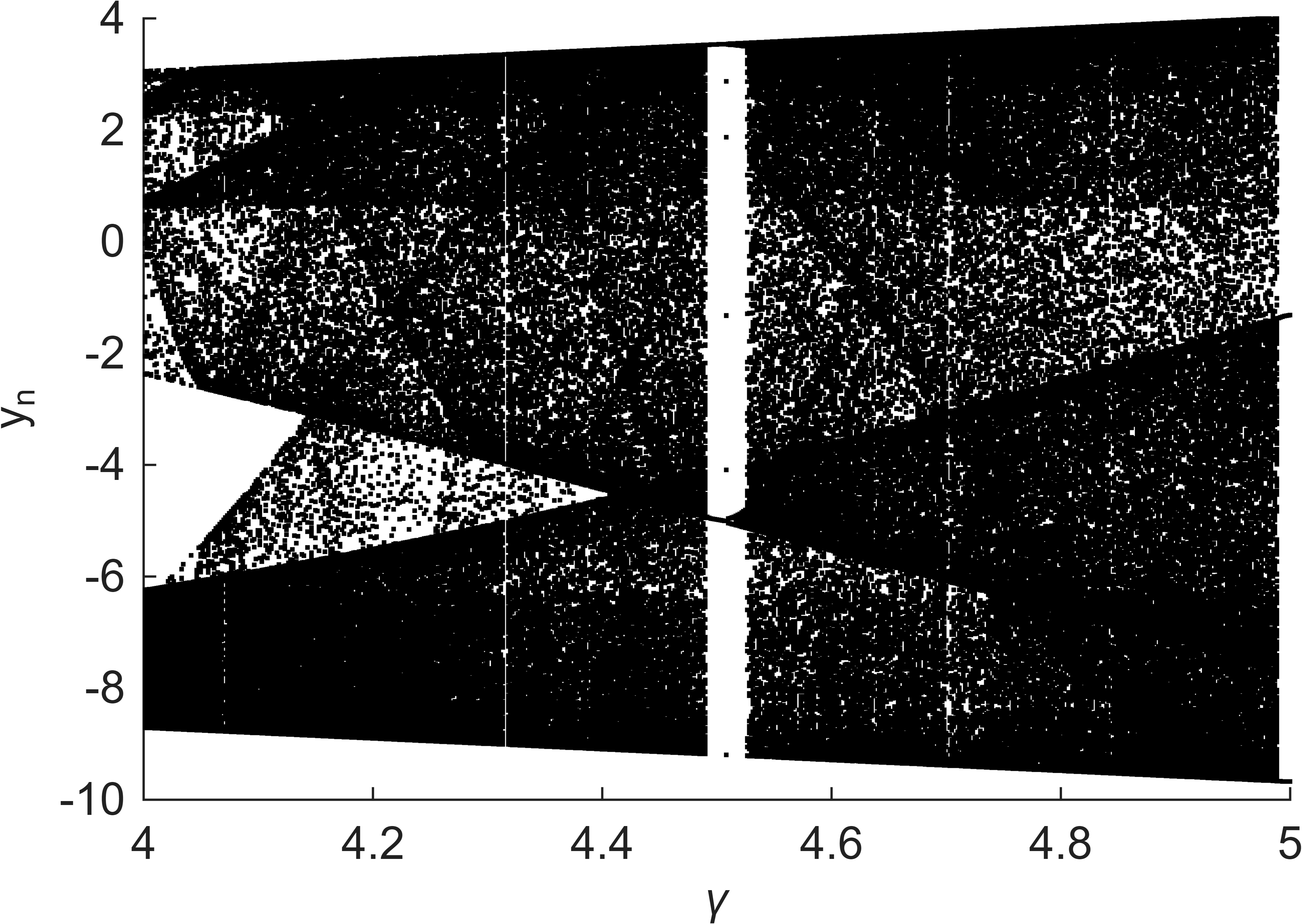}
\caption{Magnification of Figure \ref{1Dbif_gamma} related to the red rectangle.}
\label{1Dbif_gamma_zoom3}
\end{figure}
The system exhibits the typical period doubling cascade, encountered in the well known logistic map \cite{May} as well as in several discrete maps; however, it arises in a peculiar way.
\begin{itemize}
\item The system starts with periodic orbits of period $3$, $6$, $12$ up to $\gamma \sim 0.1$.
\item A large interval for $\gamma$ is associated to an attractive fixed point, depending on the value of $\gamma$ up to $\gamma \sim 2.2$.
\item A first period doubling bifurcation appears at $\gamma \sim 2.6$ and leads to a chaotic regime at $\gamma \sim 3.1$ and $y_n \in [-\frac{1}{2}, \frac{5}{2}]$.
\item The previous chaotic regime is sharply substituted by a a new 2-periodic orbit for $\gamma \sim 3.4$, followed by an analogous period doubling cascade at $\gamma \sim 3.6$ and $y_n \in [-9, -\frac{13}{2}]$.
\item From $\gamma \sim 4$ the previous chaotic regime is mixed with a different chaotic one for $y_n \in [-10, 4]$; this scenario exhibits islands of stability, for example at $\gamma \sim 4.5$.
\end{itemize}
Discontinuities exhibited in Figure \ref{1Dbif_gamma} are not totally unexpected. Since the map \eqref{1D_map}, such as \eqref{Ziegler_map_mod}, presents discontinuities of the first kind at $y_n = 2k\pi$, $k \in \bb{Z}$, we are supposed to see analogous jumps in the bifurcation diagram for $y_n$.

In discontinuous systems, multistability is a frequently observed phenomenon \cite{Pisarchik}, defined as the coexistence of different attractor and completely different behavior for a given choice on the parameters, depending solely on the initial conditions. To investigate the existence of multistable states, in Figure \ref{basins_gamma33788} we present basins of attraction of \eqref{Ziegler_map_mod} for values of the parameters corresponding to the second discontinuity shown in Figure \ref{1Dbif_gamma}. The simulations were performed using the Julia scripts provided in \cite{Wagemakers}, designed for generating basins of attraction for iterated maps and systems of ordinary differential equations. Figure \ref{basins_gamma33788} shown the projection of the basins of attraction onto the $(x, y)$ plane for $\alpha = -\tilde \alpha = -1$, $\beta = 0.1$, $\gamma = 3.3788$, with fixed $z_0 = \omega_0 = 0$ and $x_0, y_0 \in [0, 2\pi)$. The plot clearly shows the coexistence of at least three distinct attractors, identified by different colors corresponding to the asymptotic state reached by the initial conditions, confirming the presence of multistability in the system \eqref{Ziegler_map_mod}; furthermore, the basins exhibit a clear fractal structure.
\begin{figure}
\centering
\includegraphics[width=10cm]{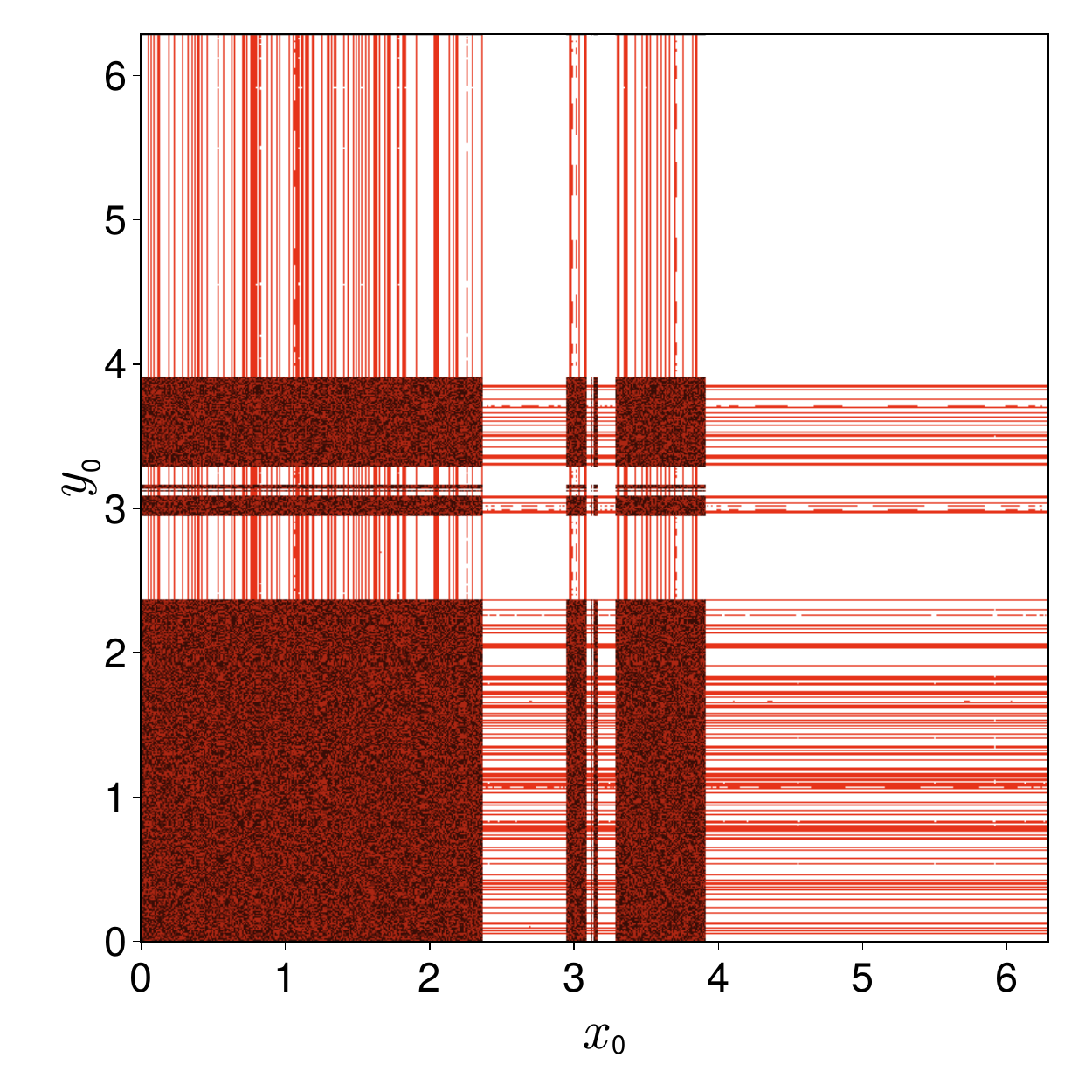}
\caption{Basins of attraction of \eqref{Ziegler_map_mod}, projection onto the $(x, y)$ plane. Parameters: $\alpha = -\tilde \alpha = -1$, $\beta = 0.1$, $\gamma = 3.3788$. Fixed initial conditions: $z_0 = \omega_0 = 0$. White, red and black regions correspond to distinct asymptotic states reached for a given set of initial condition $(x_0, y_0)$.}
\label{basins_gamma33788}
\end{figure}

Now, we perform a bifurcation analysis of the system \eqref{1D_map} in order to confirm the period doubling bifurcations shown in Figure \ref{1Dbif_gamma}. We use the software \textsc{MatContM}, a MATLAB tool devoted to bifurcation analysis of iterated smooth maps through continuation methods; see \cite{Dhooge} for a review of the native software \textsc{MatCont} (developed for ODE systems) and \cite{Neirynck} for an application of \textsc{MatContM} to a nonlinear map of economic interest.

As before, we implement the even orbit of \eqref{1D_map}, set $\tilde \alpha = - \alpha = -1$, $\beta = 0.1$ and take $\gamma$ as bifurcation parameter. In Figure \ref{1Dbif_PD} it is shown the continuation diagram for fixed points and period doubling bifurcations. We find period doubling bifurcations at $\gamma = 0.10516472$ and at $\gamma = 2.5912211$, according to the bifurcation diagram in Figure \ref{1Dbif_gamma}. It is also found a generalized period doubling bifurcation \cite{Kuznetsov} at $\gamma = 0.52673498$.
\begin{figure}
\centering
\includegraphics[width=7cm]{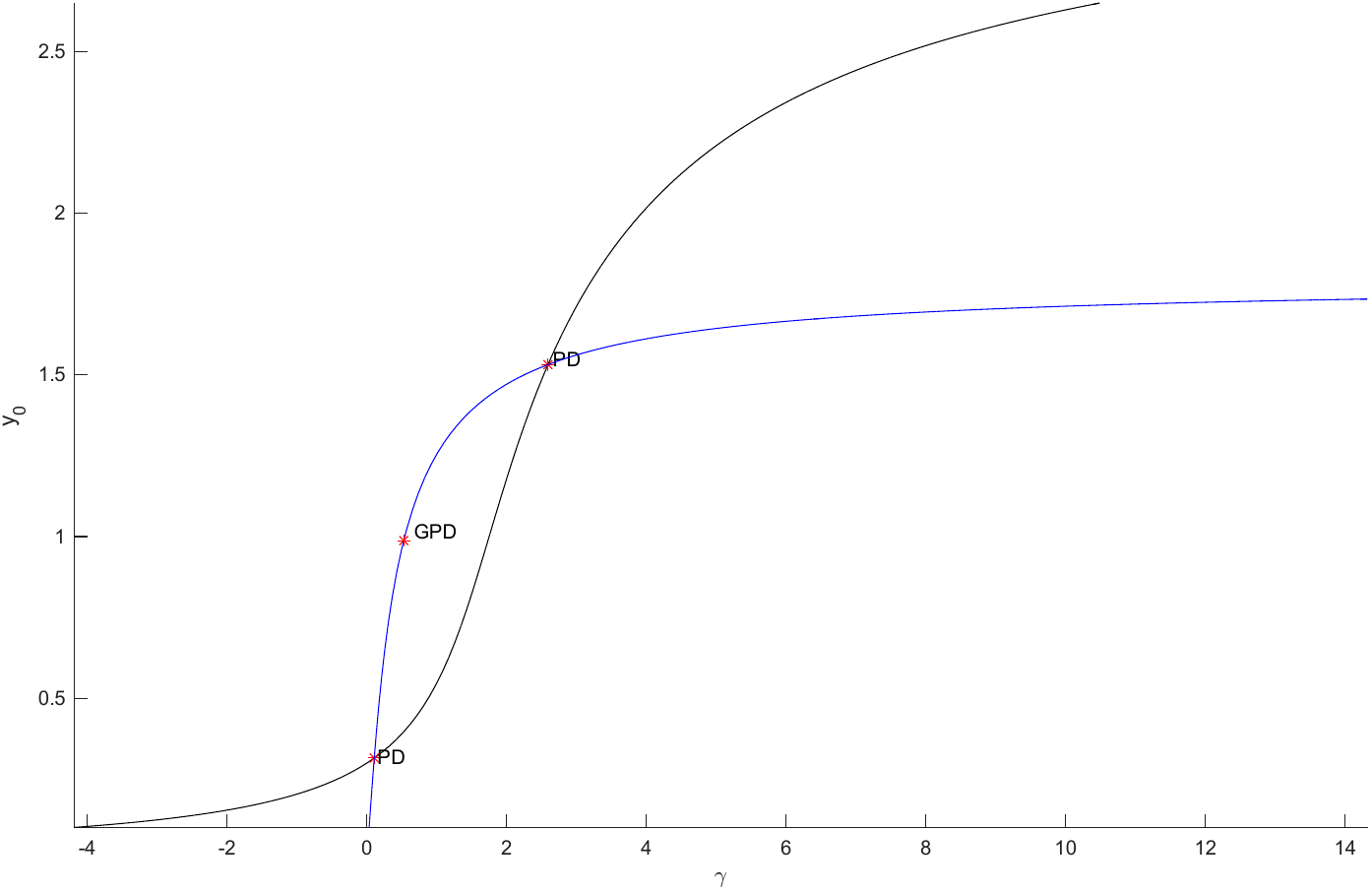}
\caption{Continuation diagram for fixed points (black) and period doubling bifurcations (blue) of the map \eqref{1D_map}. $\tilde \alpha = -\alpha = 1$, $\beta = 0.1$. Two period doubling bifurcations (PD) are found at $\gamma = 0.10516472$ and $\gamma = 2.5912211$. A generalized period doubling bifurcation (GPD) is found at $\gamma = 0.52673498$.}
\label{1Dbif_PD}
\end{figure}

In Figure \ref{1Dlyap} we present Lyapunov exponents for the system \eqref{1D_map} for $\alpha = -\tilde \alpha = -1$, $\beta = 0.1$, $\gamma \in [0,5]$; we implement the even orbit for \eqref{1D_map} with initial condition $y_0 = 0.1$. Negative Lyapunov exponents are found for values of $\gamma$, according to the stability islands shown in Figure \ref{1Dbif_gamma} - \ref{1Dbif_gamma_zoom3}.
\begin{figure}
\centering
\includegraphics[width=7cm]{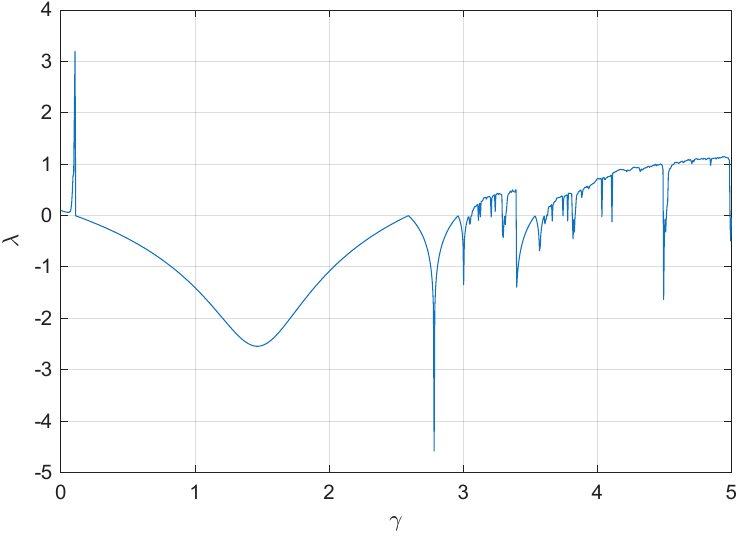}
\caption{Lyapunov exponents of the map \eqref{1D_map} for $y_0 = 0.1$ with respect to $\gamma \in [0,5]$. Parameters: $\tilde \alpha = -\alpha = 1$, $\beta = 0.1$, $y_0 = 0.1$.}
\label{1Dlyap}
\end{figure}

Numerical results presented in Figure \ref{1Dbif_gamma} - \ref{1Dbif_gamma_zoom3}, \ref{1Dbif_PD} - \ref{1Dlyap} refer to the 1D map \eqref{1D_map}, whose dynamics is equivalent to the four-dimensional \eqref{Ziegler_map_mod} when $\alpha = -\tilde \alpha$. Now, we present 2D bifurcation diagrams for the full system \eqref{Ziegler_map_mod} and spectra of Lyapunov exponents, that were generated using a dedicated Python implementation. For a comprehensive overview of a recent high-performance numerical tool for bifurcation analysis, we refer to \cite{Rybin}. \\We perform bifurcation analysis for three sets of parameters. For all the following simulations we consider $n_{it} = 5.000$ iterations by discarding the first $n_{tran} = 500$ ones and discretizing the parameters range into $1.000$ values. The color bar indicates the maximum Lyapunov exponents of \eqref{Ziegler_map_mod} for a given pair of parameters.

In Figure \ref{2Dbif_beta_gamma} it is shown the 2D bifurcation diagram $(\beta, \gamma)$ for $\alpha = -\tilde \alpha = -1$, $\beta \in [0, 5]$, $\gamma \in [0, 5]$. In Figure \ref{lyap_spectrum_beta_gamma} the Lyapunov spectrum for same values of $\alpha$, $\tilde \alpha$ and $\beta = 0.1$ is presented, by varying $\gamma$ in the interval $[0, 5]$.
\begin{figure}
\centering
\includegraphics[width=12cm]{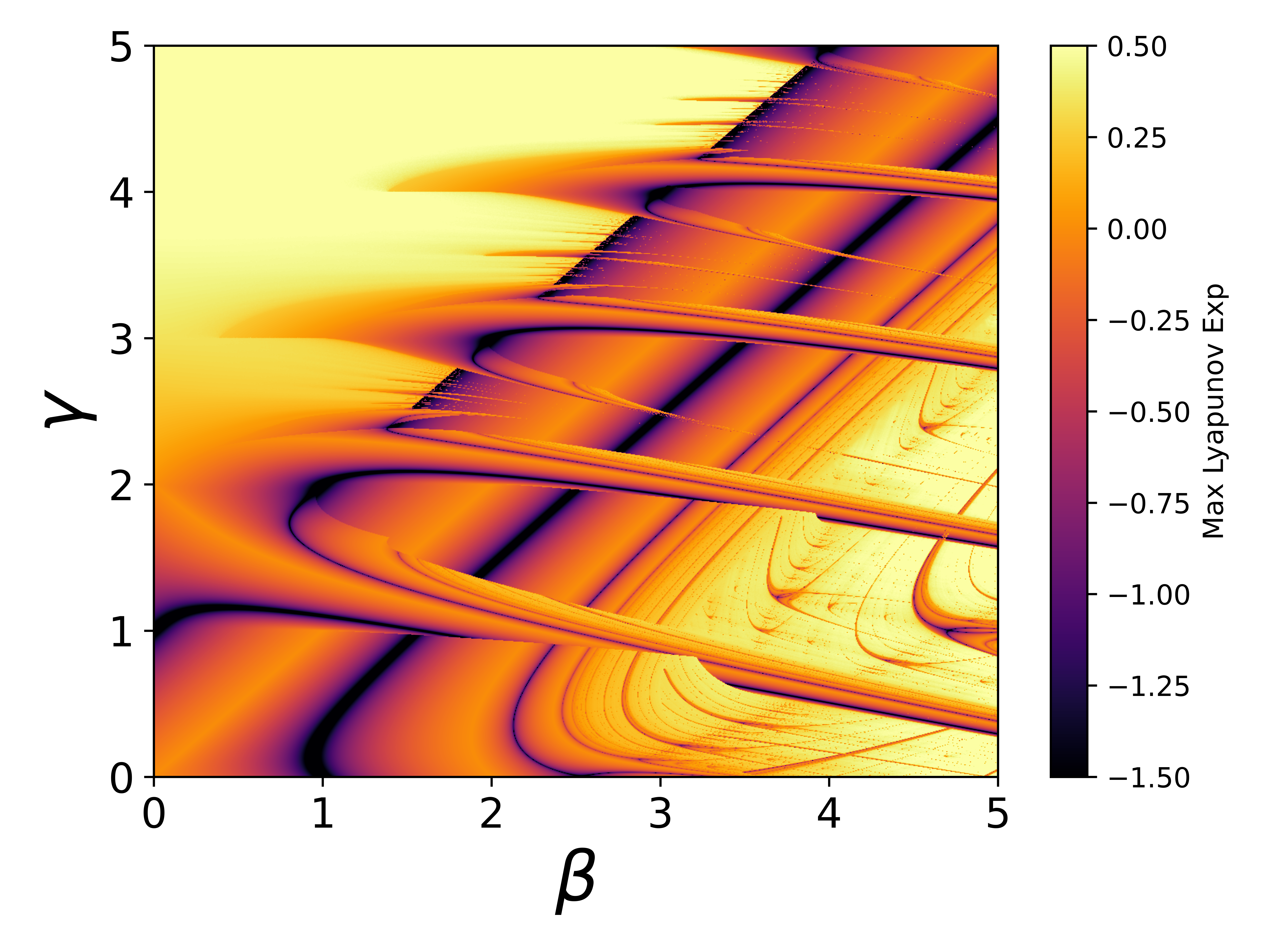}
\caption{2D bifurcation diagram $(\beta, \gamma)$ for the map \eqref{Ziegler_map_mod}. Parameters: $\alpha = -\tilde \alpha = -1$.}
\label{2Dbif_beta_gamma}
\end{figure}
\begin{figure}
\centering
\includegraphics[width=8cm]{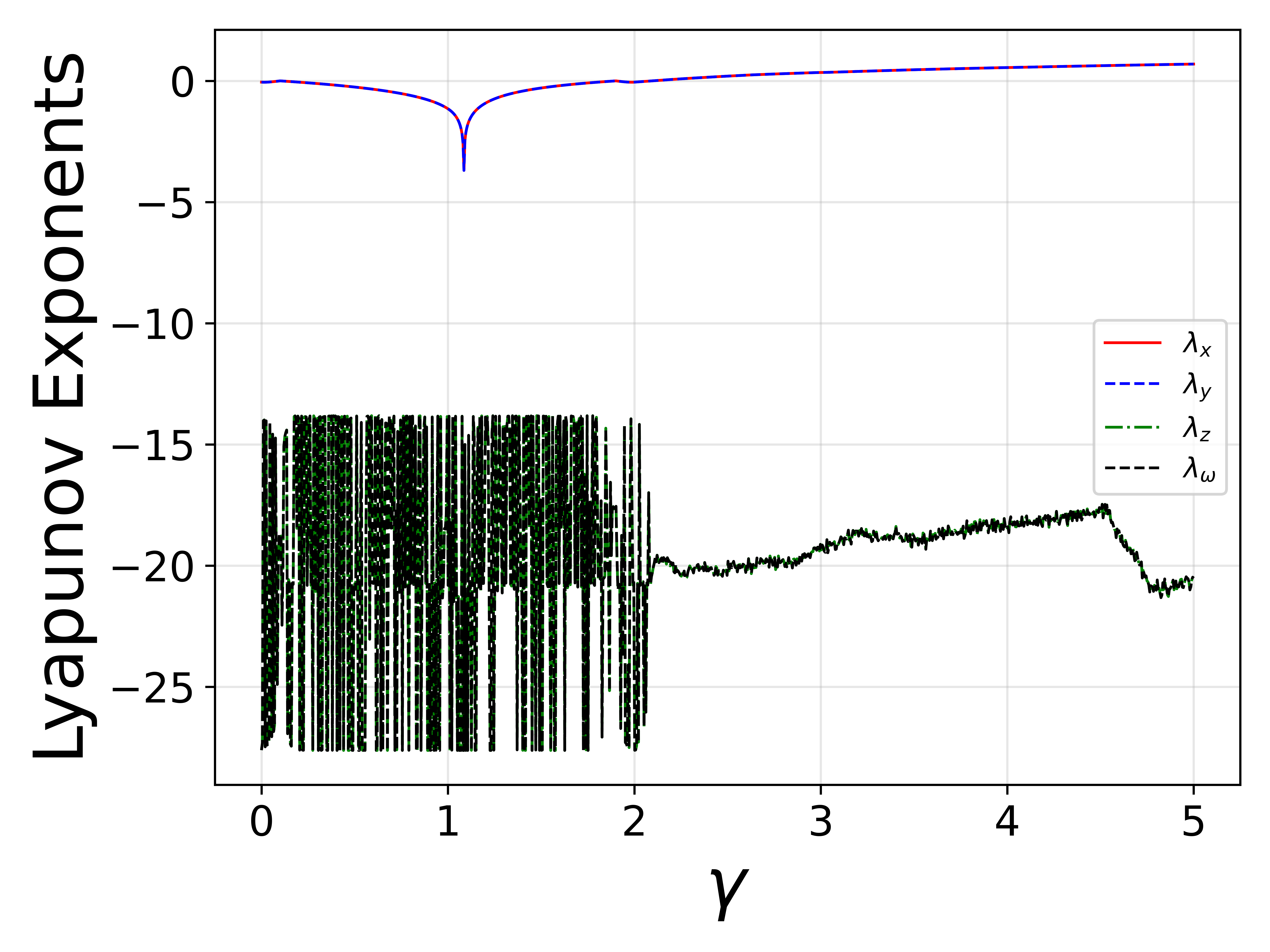}
\caption{Lyapunov exponents spectrum for the map \eqref{Ziegler_map_mod}. Parameters: $\alpha = -\tilde \alpha = -1$, $\beta = 0.1$.}
\label{lyap_spectrum_beta_gamma}
\end{figure}

In Figure \ref{2Dbif_alfa_gamma} it is shown the 2D bifurcation diagram $(\alpha, \gamma)$ for $\tilde \alpha = 2$, $\beta = 0.8$, $\alpha \in [-5, 0]$, $\gamma \in [0, 10]$. In Figure \ref{lyap_spectrum_alfa_gamma} the Lyapunov spectrum for same values of $\tilde \alpha$, $\beta$ and $\alpha = -2.5$ is presented, by varying $\gamma$ in the interval $[0, 10]$. These results are comprehensive of the parameters choice for the confined strange attractor shown in Figure \ref{attractor_xz_tot} - \ref{attractor_yw_tot}.
\begin{figure}
\centering
\includegraphics[width=12cm]{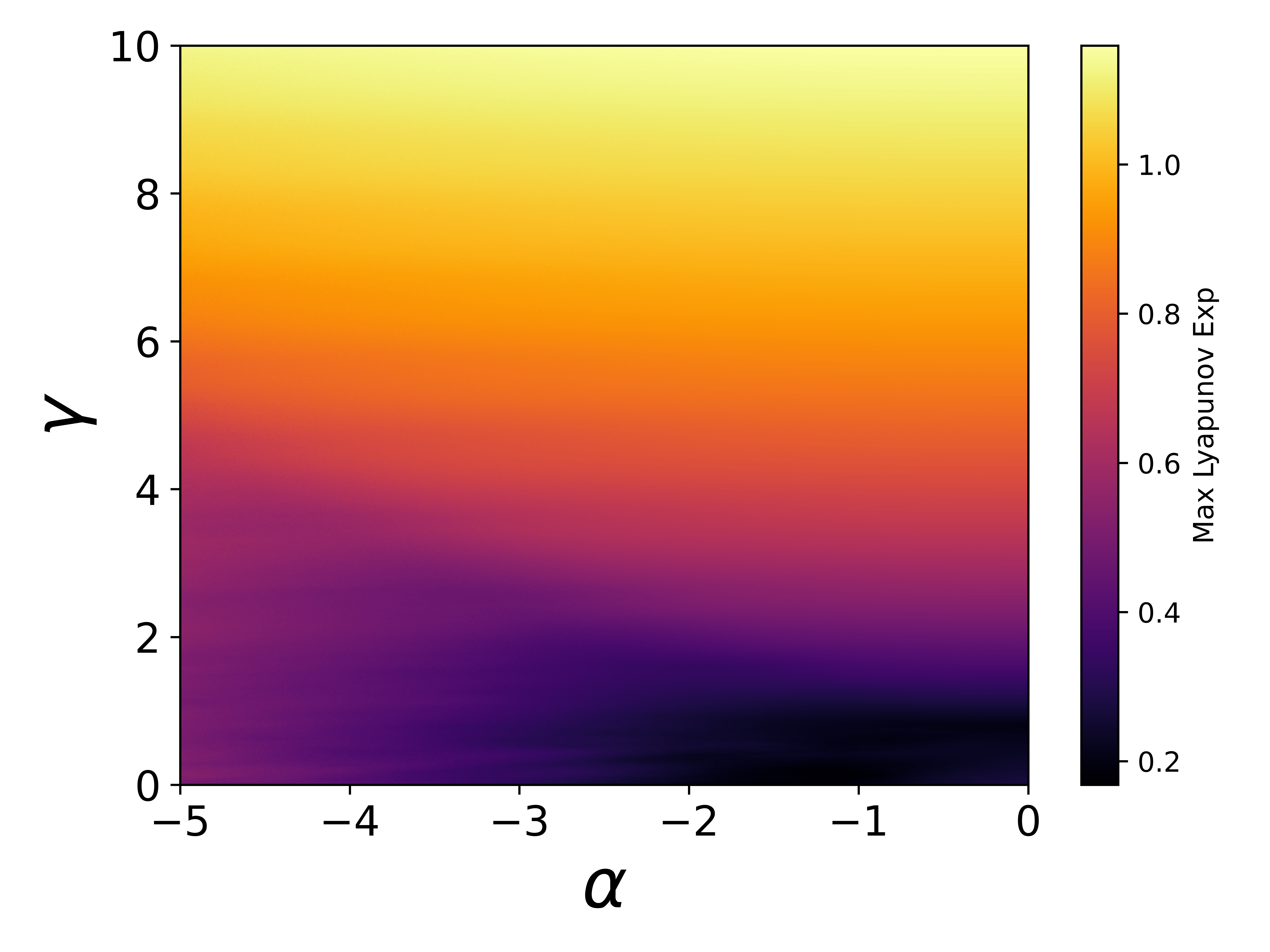}
\caption{2D bifurcation diagram $(\alpha, \gamma)$ for the map \eqref{Ziegler_map_mod}. Parameters: $\tilde \alpha = 2$, $\beta = 0.8$.}
\label{2Dbif_alfa_gamma}
\end{figure}
\begin{figure}
\centering
\includegraphics[width=8cm]{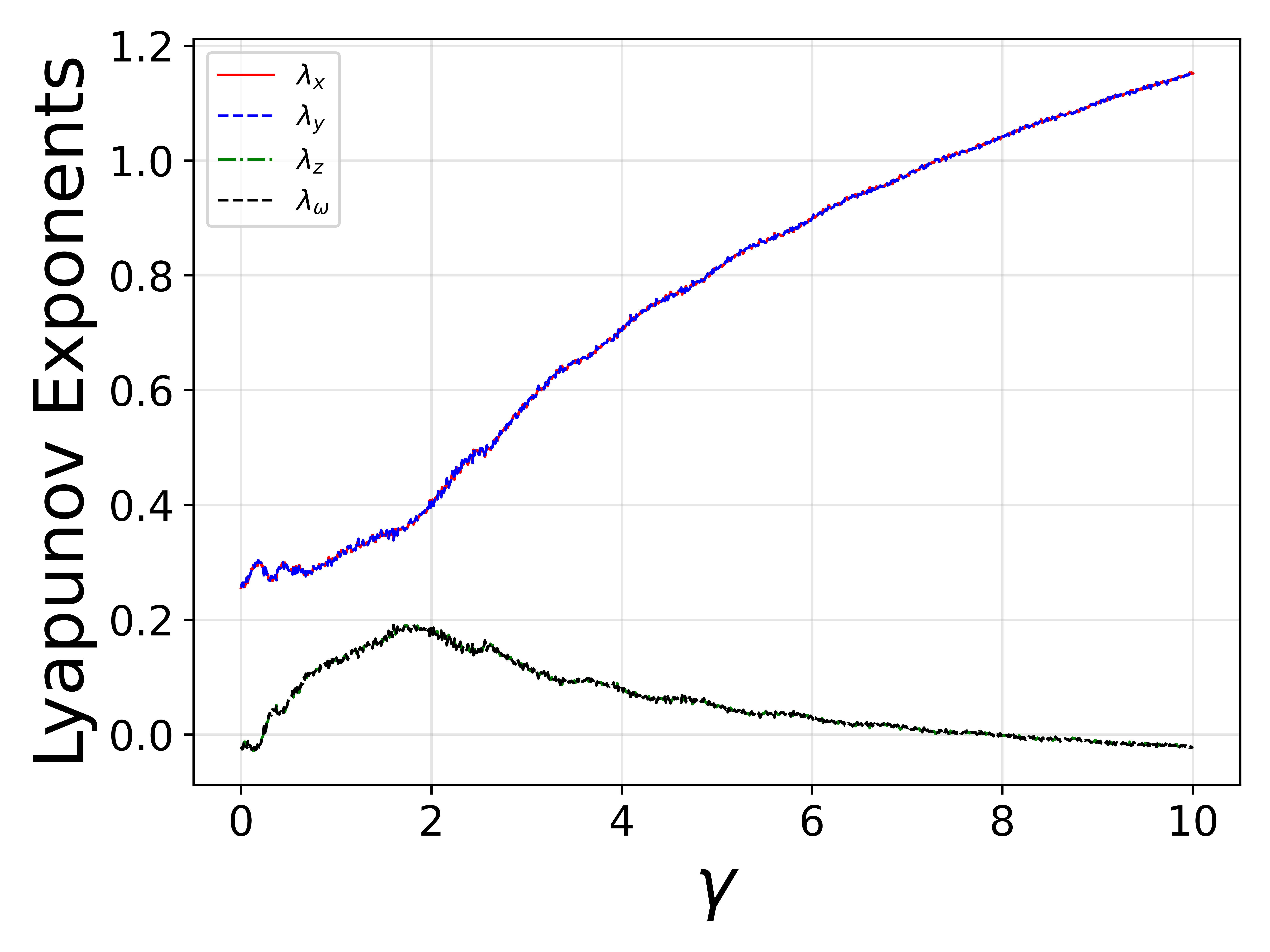}
\caption{Lyapunov exponents spectrum for the map \eqref{Ziegler_map_mod}. Parameters: $\tilde \alpha = 2$, $\beta = 0.8$, $\alpha = -2.5$.}
\label{lyap_spectrum_alfa_gamma}
\end{figure}

In Figure \ref{2Dbif_alfa_beta} it is shown the 2D bifurcation diagram $(\alpha, \beta)$ for $\tilde \alpha = 2$, $\gamma = 4$, $\alpha \in [-1, 1]$, $\beta \in [-1, 1]$. In Figure \ref{lyap_spectrum_alfa_beta} the Lyapunov spectrum for same values of $\tilde \alpha$, $\gamma$ and $\alpha = 0$ is presented, by varying $\beta$ in the interval $[-1, 1]$.
\begin{figure}
\centering
\includegraphics[width=12cm]{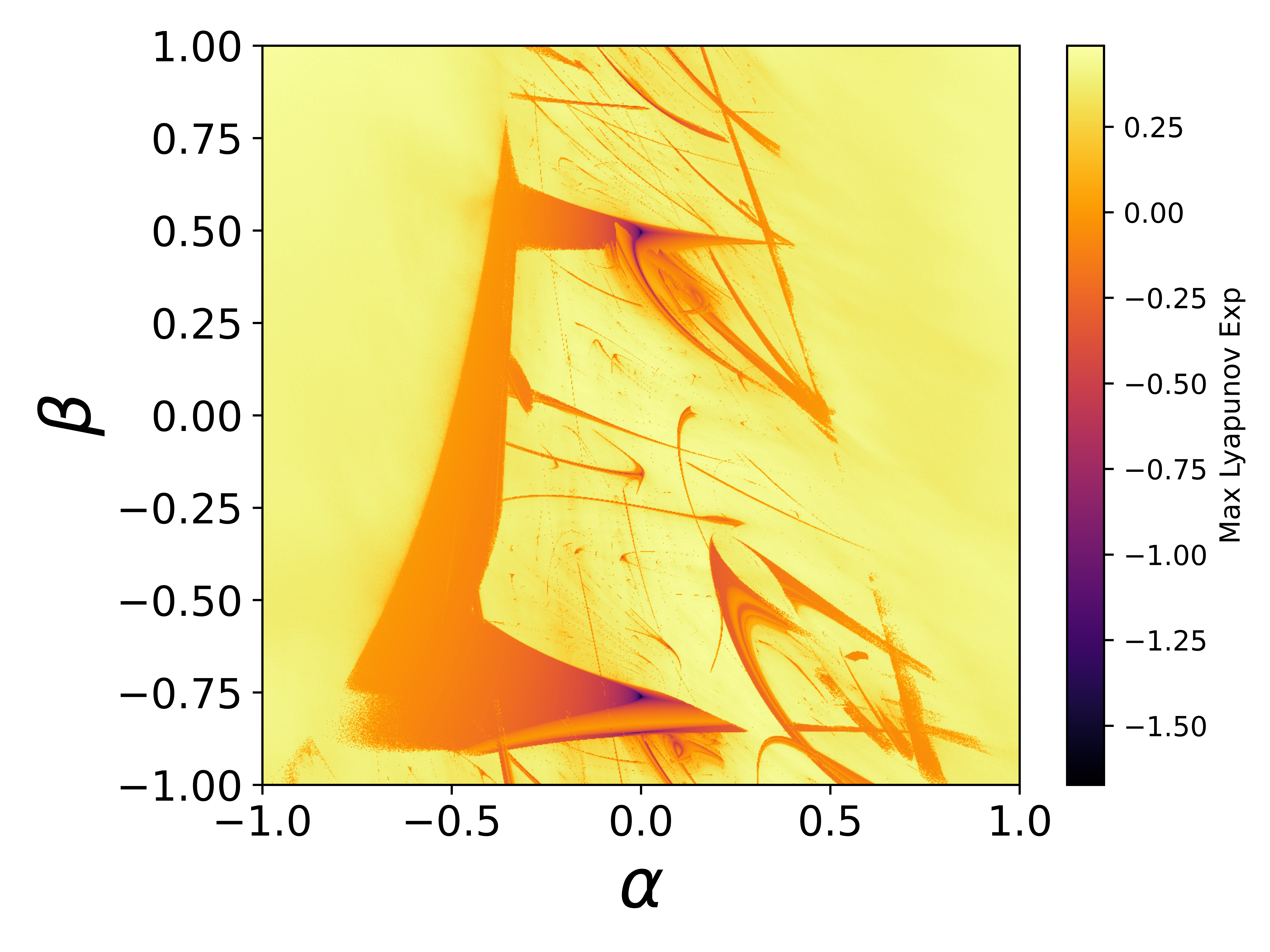}
\caption{2D bifurcation diagram $(\alpha, \beta)$ for the map \eqref{Ziegler_map_mod}. Parameters: $\tilde \alpha = 2$, $\gamma = 4$.}
\label{2Dbif_alfa_beta}
\end{figure}
\begin{figure}
\centering
\includegraphics[width=8cm]{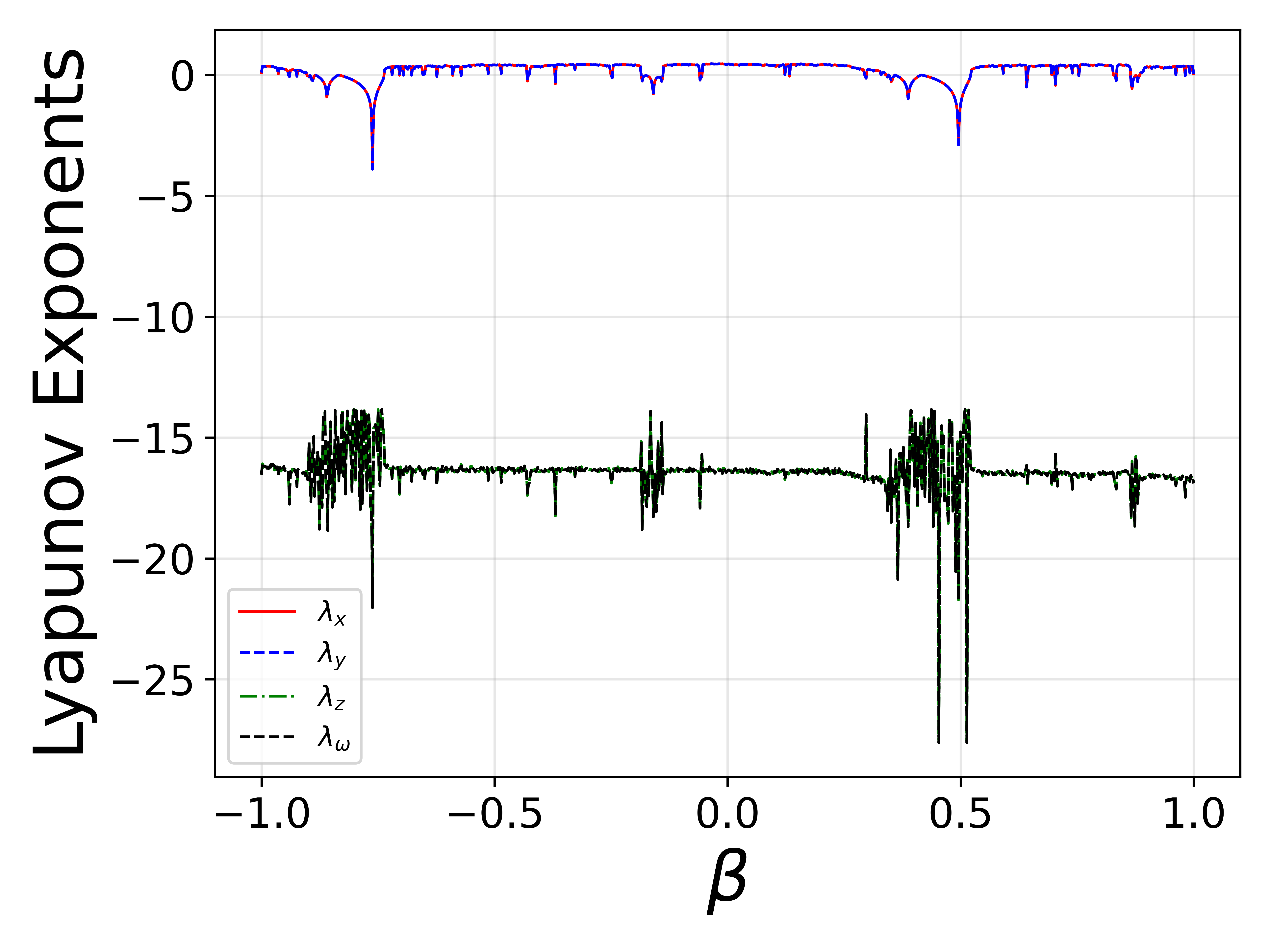}
\caption{Lyapunov exponents spectrumfor the map \eqref{Ziegler_map_mod}. Parameters: $\tilde \alpha = 2$, $\gamma = 4$, $\alpha = 0$.}
\label{lyap_spectrum_alfa_beta}
\end{figure}

A few remarks are in order. Firstly, notice that the Jacobian of \eqref{Ziegler_map_mod} exhibits a symmetry on its eigenvalues such that Lyapunov exponents are equal in pairs, as evident from the spectra shown in Figure \ref{lyap_spectrum_beta_gamma}, \ref{lyap_spectrum_alfa_gamma}, \ref{lyap_spectrum_alfa_beta}. Furthermore, for the parameters set selected in Figure \ref{lyap_spectrum_alfa_gamma}, we have a positive sum of the Lyapunov exponents, that is incompatible with a global strange attractor; however, a confined strange attractor can still emerge in specific projections, such as the one that is observed in Figure \ref{attractor_xz_tot} - \ref{attractor_yw_tot}. Finally, for all the previous simulations wide sets of parameters exhibit two positive Lyapunov exponents, leading to hyperchaos \cite{Kapitaniak, Stankevich} for the system \eqref{Ziegler_map_mod}.

\section{Conclusions}\label{sec_conclusions}
In this work we analyzed a family of four-dimensional discrete dynamical systems defined through polynomial and periodic functions of one coordinate, proving that the map is chaotic in the sense of Li-Yorke under suitable requests on the parameters. Relevance of Theorem \ref{theo_CD1} - \ref{theo_CD3} lies on the connection that they establish between chaotic behavior of these discrete maps and respective continuous systems, when discretized in a proper way. Discrete maps satisfying hypothesis of Theorem \ref{theo_CD1} - \ref{theo_CD3} can be seen as discrete versions of several continuous dynamical systems, that can be integrable, quasi-integrable or chaotic; regardless of the regularity of the motion exhibited by these systems, their discrete version, as considered in this work, is always Li-Yorke chaotic. This fact is open to some interpretations. In particular, it shows how the necessary conditions for the emerging of chaos strongly depend on the choice between a discrete or continuous evolution step, at least within the framework considered in this work; on the other hand, it is a further confirmation of the generality of the definition given by Li-Yorke, that is sufficiently weak to be satisfied by a large class of dynamical systems.

Despite imposing a number of constraints, hypothesis of Theorem \ref{theo_CD1} - \ref{theo_CD3} are sufficiently general to be applied in modeling of interest in several fields, from physics to biology to economics.

Chaotic behavior of the modular definition of the map has been numerically confirmed by identifying a confined strange attractor for general choices on parameters and initial conditions. The one-dimensional bifurcation diagram of the map exhibits a fractal structure analogous to the logistic map, with a singular rising of the period doubling cascade. Sensitive dependence of the system on initial data and rising of positive Lyapunov exponents are confirmed by two-dimensional bifurcation diagrams and the Lyapunov exponents spectra. Furthermore, the basins of attraction of the system confirm the presence of multistable states.

Results presented in this work naturally suggest a number of questions.

They might be taken into consideration extensions of Theorem \ref{theo_CD1} - \ref{theo_CD3} to the most large generalizations of the problem; in particular, we might think to generalize Theorem \ref{theo_CD1} - \ref{theo_CD3} for polynomial function of $x_n$ and arbitrary periodic functions of $x_n$. Furthermore, we may think about a further extension to analogous discrete maps defined in $\bb{R}^{2n}$, $n \in \bb{N}$.

In addition to the previous ones, further proofs might be provided in order to generalize Theorem \ref{theo_CD1} - \ref{theo_CD3} for standard discretization procedures of the system; conversely, we might derive continuous systems that are discretized in the form \eqref{Ziegler_map} through usual discretization methods.

From a purely theoretical point of view, the comment presented in Section \ref{sec_theo_comment} could be an interesting starting point for  studies at the intersection of chaos theory and KAM theory.

The rising of a strange attractor for the system leads to some questions. Firstly, we might numerically compute the Hausdorff dimension \cite{Hausdorff} of the attractor, in order to have a confirmation of its fractal dimension. Secondly, it is not immediate to provide an explanation for certain dynamical properties, such as the rising of the attractor only on certain projection planes or the peculiar regular orbits rising in the symmetric case. Finally, the bifurcation analysis performed in this work may be extended to larger sets of parameters, in order to provide a complete scenario for stability and transition to chaos for the system.

We conclude posing a last conceptual question concerning theoretical and practical meaning of the fixed point discretization we used in this work. We stress the fact that the discretization procedure we considered here does not belong to the class of standard algorithms to approximate ODEs with discrete maps; however, it could be possible to provide a formal framework for this alternative discretization procedure such that the discrete version can give informations about the behavior of the associate continuous system. Furthermore, we might ask ourselves what analytical and topological properties a continuous system and its fixed point discretization share.

We leave all these questions for future developments of this work.
\\ \\ \\ \textbf{CRediT authorship contribution statement}

Conceptualization, S.D. and V. C.; Methodology, S.D. and V. C.; Software, S.D.; Validation, V. C.; Formal analysis, S.D.; Writing--original draft, S.D.; Writing--review and editing, V. C.; Visualization, S.D.; Supervision, V. C.; Project administration, V. C.; Funding acquisition, V. C. All authors have read and agreed to the published version of the manuscript.
\\ \\ \textbf{Declaration of competing interests}

The authors declare no conflicts of interest.
\\ \\ \textbf{Funding}

This research was funded by the University of Ferrara, FIRD 2024.
\\ \\ \textbf{Data availability}

The data that supports the funding of this study are available within the article.
\\ \\ \textbf{Declaration of generative AI and AI-assisted technologies in the manuscript preparation process}

During the preparation of this work the authors used Google Gemini (Large Language Model) in order to write the Python code for reproducing the 2D bifurcation diagrams and the Lyapunov exponents spectra (Figure \ref{2Dbif_beta_gamma} - \ref{lyap_spectrum_alfa_beta}). After using this tool, the authors reviewed and edited the content as needed and take full responsibility for the content of the published article.
\\ \\ \textbf{Acknowledgments}

The first author is grateful to prof. Davide Liessi (University of Udine) for his insightful lectures on the MATLAB tool \textsc{MatCont}. The authors also thank prof. Michele Miranda (University of Ferrara) for helpful discussions regarding the initial flawed proof of Theorem \ref{theo_CD1}. The authors are grateful to the anonymous reviewers for their helpful comments and suggestions.

\bibliographystyle{elsarticle-num}
\bibliography{discrete_biblio_accepted}

\end{document}